\documentclass[12pt]{amsart}
\usepackage{color}
\usepackage{graphicx,epsf}
\usepackage{subfigure}
\usepackage{amsmath,amscd}
\newtheorem{theorem}{Theorem}[section]
\newtheorem{lemma}[theorem]{Lemma}
\newtheorem{corollary}[theorem]{Corollary}

\theoremstyle{definition}
\newtheorem{definition}[theorem]{Definition}
\newtheorem{example}[theorem]{Example}

\theoremstyle{remark}
\newtheorem{remark}[theorem]{Remark}
\numberwithin{equation}{section}
\newcommand{\abs}[1]{\left\vert#1\right\vert}
\begin{document}
\title[Singular perturbation, $PT$-symmetric families]{Singular perturbation of polynomial potentials
with applications to
$PT$-symmetric families}
\date{\today}
\author{Alexandre Eremenko}
\address{Purdue University, West Lafayette, IN 47907}
\email{eremenko@math.purdue.edu}
\author{Andrei Gabrielov}
\address{Purdue University, West Lafayette, IN 47907}
\email{agabriel@math.purdue.edu}
\thanks{A. Eremenko was supported by NSF
grant DMS--0555279, and by the Humboldt Foundation;
A. Gabrielov was supported by
NSF grant DMS--0801050.}
\def\eps{\varepsilon}
\def\R{\mathbf{R}}
\def\C{\mathbf{C}}
\def\bC{{\mathbf{\overline{C}}}}
\def\mod{\mathrm{mod}\ }
\def\const{\mathrm{const}}
\def\St{\mathrm{St}}
\def\a{\mathbf{a}}
\def\re{\mathrm{Re}\,}
\def\im{\mathrm{Im}\,}
\begin{abstract}
We discuss eigenvalue problems of the form
$-w^{\prime\prime}+Pw=\lambda w$ with complex
polynomial potential 
$P(z)=tz^d+\ldots$,
where $t$ is a parameter,
with zero boundary conditions at infinity on
two rays in the complex plane.
In the first part of the paper we
give sufficient conditions for
continuity of the spectrum at $t=0$.
In the second part we apply these results 
to the study of topology
and geometry of the real
spectral loci of $PT$-symmetric families
with $P$ of degree $3$ and $4$, and prove several related
results on the location of zeros of their eigenfunctions.

\end{abstract}
\maketitle

\section{Introduction}
We consider eigenvalue problems
\begin{equation}\label{problem}
-w^{\prime\prime}+P(z,\a)w=\lambda w,\quad
y(z)\to 0\;\mathrm{as}\; z\to\infty,\; z\in L_1\cup L_2.
\end{equation}
Here $P$ is a polynomial in the independent variable $z$,
which depends on a parameter $\a$,
and $L_1,\,L_2$ are two rays in the complex plane.
The set of all pairs $(\a,\lambda)$ such that $\lambda$
is an eigenvalue of (\ref{problem}) is called
the {\em spectral locus}.

Such eigenvalue problems were considered for the first
time in full generality
by Sibuya \cite{Sibuya} and Bakken \cite{Bakken}.
Sibuya proved that under certain
conditions on $L_1,\ L_2$ and the leading
coefficient of $P$, there exists an infinite sequence
of eigenvalues tending to infinity. 
If
\begin{equation}\label{P}
P(z,\a)=z^d+a_{d-1}z^{d-1}+\ldots+a_1z,
\end{equation}
where $\a=(a_1,...,a_{d-1}),$ then the spectral locus, which is
the set of all $(\a,\lambda)\in \C^d$ such that $\lambda$ is an
eigenvalue of \ref{problem}, is described by an equation $F(\a,\lambda)=0.$
Here $F(\a,\lambda)$ is an entire function of $d$ variables, called
the {\em spectral determinant.} 
So the spectral locus of (\ref{problem}),
(\ref{P}) is an analytic hypersurface
in $\C^d$. It is smooth \cite{Bakken} and
connected for $d\geq 3$ \cite{Per,Habsch}. 

In the first part of this paper
we study what happens to the
eigenvalues and eigenfunctions
when the leading coefficient of $P$
tends to zero.

Bender and Wu \cite{BW}
studied the quartic oscillator as a perturbation
of the harmonic oscillator:
\begin{equation}\label{bw}
-w^{\prime\prime}+(\eps z^4+z^2)w=\lambda w,\quad
w(\pm\infty)=0.
\end{equation}
Here and in what follows $w(\pm\infty)=0$
means that the boundary conditions are
imposed on the positive
and negative rays of the real line.
It has been known for long time
that the eigenvalues of (\ref{bw})
converge as $\eps\to 0+$ to the eigenvalues
of the same problem with $\eps=0$, but they are not
analytic functions of $\eps$ at $\eps=0$
(perturbation series diverge).
To investigate this phenomenon, Bender and Wu
considered complex values of $\eps$
and studied the analytic continuation of the eigenvalues
as functions of $\eps$ in the complex plane.
Their main findings
can be stated as follows: the spectral locus
of the problem (\ref{bw}) consists of exactly
two connected components; for $\eps\neq 0$,
the only singularities of eigenvalues as functions
of $\eps$ are algebraic branch points.
These statements were rigorously proved in \cite{EG2}.
Discoveries of Bender and Wu generated large
literature
in physics and mathematics. For a comprehensive
exposition of the early rigorous results
we refer to \cite{Simon}.

To perform analytic continuation
of eigenvalues of (\ref{bw}) and similar problems
for complex parameters,
one has to rotate the normalization rays where
the boundary conditions are imposed.
One of the  early
papers in the physics literature that emphasized this
point was \cite{BT}. Thus physicists were led
to problem (\ref{problem}), previously studied
only for its intrinsic mathematical interest.

An interesting phenomenon was discovered by Bessis and Zinn-Justin.
For the boundary value problem
$$-w^{\prime\prime}+iz^3w=\lambda w,\quad w(\pm\infty)=0,$$
they found by numerical computation
that the spectrum is real.
This is called the Bessis and Zinn-Justin conjecture
(see, for example, historical remark in
\cite{BBo}). 
This conjecture
was later proved by Dorey, Dunning and Tateo \cite{DT0,DT1}
with a remarkable argument which they call the 
ODE-IM correspondence, see their survey \cite{DT2}.
Shin \cite{Shin} extended this result to potentials
\begin{equation}\label{tr}
-w^{\prime\prime}+(iz^3+ia z)w=\lambda w,
\quad w(\pm\infty)=0,
\end{equation}
with $a\geq 0$.

These results and conjectures generated extensive research
on the so-called $PT$-symmetric boundary value problems.
$PT$-symmetry means a symmetry of the potential and
of the boundary conditions with respect to the reflection in
the imaginary line $z\mapsto -\overline{z}$.
$PT$ stands for ``parity and time reversal''.

It turns out that the spectral determinant of a
$PT$-symmetric
problem is a real entire function of $\lambda$, so the
set of eigenvalues is invariant under complex conjugation.
In contrast to Hermitian problems where the
eigenvalues are always real,
the eigenvalues of a $PT$-symmetric problem
can be real for some values of parameters, but for other
values of parameters some eigenvalues may be complex.
So we can see the ``level crossing'' (collision of real
eigenvalues) in real analytic families of
$PT$-symmetric operators,
the phenomenon which is impossible in the
families of Hermitian differential
operators with polynomial coefficients.

In this paper,
we first consider the general problem
(\ref{problem}) and the limit behavior of its eigenvalues and
eigenfunctions when 
\begin{equation}
\label{aab}
P(z)=t z^d+a_m z^m+p(z),
\end{equation}
with $d>m>\deg p$, as $t\to 0$, while the coefficients
of $p$ are restricted to a compact set and $a_m$
does not approach zero.
Then we apply our general results to certain families
of $PT$-symmetric potentials of degrees $3$ and $4$, and
prove some conjectures made by several authors
on the basis of numerical evidence.

In particular, our results for the $PT$-symmetric
cubic (\ref{tr}) imply that no eigenvalue can be
analytically continued along the negative $a$-axis,
and the obstacle to this continuation is
a branch point where eigenvalues collide. 

Another result is the correspondence between the
natural ordering of real
eigenvalues of (\ref{tr}) for $a\geq 0$ and the number
of zeros of eigenfunctions that do not lie on
the $PT$-symmetry axis, conjectured
by Trinh in \cite{Trinh2}. This correspondence is similar to
that given by the Sturm--Liouville theory for Hermitian
boundary value problems.

A different approach to counting zeros of eigenfunctions
is proposed in \cite{GMM}, where the authors prove that
for $a$ large enough, the $n$-th eigenfunction has $n$
zeros in a certain explicitly described
region in the complex plane.

The plan of the paper is the following.
In Section 2 we prove a general theorem
on the continuity of discrete spectrum
at $t=0$ for potentials of the form (\ref{aab}), 
with boundary conditions on two given rays.
Previously such problems were studied using the
perturbation theory of linear operators in \cite{GGS,Simon,CGM}.
Our method is different, it is based on analytic theory of differential
equations.

Verification of conditions of our general result in Section 2
is non-trivial, and we dedicate the entire Section 3 to this.
The question is reduced to the study of Stokes complexes
of binomial potentials $Q(z)=t z^d+c z^m,\; d>m$, which is a problem
of independent interest, so we include more detail than it is necessary for
our applications.
The Stokes complex is the union of curves,
starting at the zeros of $Q$, on which
$Q(z)\,dz^2<0$, so they are vertical trajectories
of a quadratic differential. 
Stokes complexes occur in many questions about asymptotic
behavior of solutions of equations (\ref{problem}).
Our study permits us to make conclusions on
the behavior, as $t\to 0$, of the Stokes complexes of potentials
$P(z)=t z^d+a_m(t)z^m+p_{t}(z)$ where $a_m(t)\to c\ne 0$ 
and $p_{t}$ is a family of polynomials
of degree $m-1$ with bounded coefficients.
We mention here \cite{Mas} where a topological classification
of Stokes complexes for polynomials of degree $3$ is given.

In the rest of the paper we apply these results to
problems with $PT$-symmetry. In Section 4, we consider the
$PT$-symmetric cubic family (\ref{tr})
with real $a$ and $\lambda$.
We prove that the intersection of the spectral locus
with the real $(a,\lambda)$-plane consists of
disjoint non-singular analytic curves
$\Gamma_n,\; n\geq 0$, the
fact previously known from numerical computation
\cite{DT,Trinh,Handy}. Moreover, we prove that
the eigenfunctions corresponding
to $(a,\lambda)\in \Gamma_n$
have exactly $2n$ zeros outside the imaginary line.
(They have infinitely many zeros on the imaginary line).
Furthermore, using the result of Shin on reality
of eigenvalues, we study the shape
and relative location of these curves
$\Gamma_n$ in the $(a,\lambda)$-plane
and show that $a\to+\infty$ on both ends of $\Gamma_n$, and
that for $a\geq 0$, $\Gamma_n$ consists of
graphs of two functions, that lie below the graphs
of functions constituting $\Gamma_{n+1}$.

This gives $PT$--analog of the familiar
fact for Hermitian boundary value problems that
``$n$-th eigenfunction has $n$ real zeros''; in our case
we count zeros belonging to a certain well-defined set
in the complex plane.
This result proves rigorously what
can be seen in numerical computations of zeros of eigenfunctions
by Bender, Boettcher and Savage~\cite{BBS}.

The result of Section 4 also gives a contribution to
a problem raised by Hellerstein and Rossi \cite{HC}:
describe the differential equations
\begin{equation}\label{hr}
y^{\prime\prime}+Py=0
\end{equation}
with polynomial coefficient $P$ which have a solution whose
all zeros are real. For polynomials of degree $3$,
all such equations are parametrized
by our curve $\Gamma_0$, and equations having solutions with
exactly $2n$ non-real zeros are parametrized by $\Gamma_n$.

The arguments in Section 4 use our 
parametrization of the spectral loci
from \cite{EG1,EG2} combined with the singular perturbation
results of Sections 2 and 3.
These perturbation results allow us
to degenerate the cubic potential to a quadratic one
(harmonic oscillator) and to make topological
conclusions based on the ordinary Sturm-Liouville theory.

Next we apply similar methods
to two families of $PT$-symmetric quartics
\begin{equation}\label{BB}
-w^{\prime\prime}+(z^4+az^2+ic z)w=\lambda w,\quad
w(\pm\infty)=0.
\end{equation}
and
\begin{eqnarray}\label{BS1}
w^{\prime\prime}+(z^4+2az^2+2imz+\lambda)w=0,\\
\label{BS2}
w(re^{\theta})\to 0,\;\mathrm{as}\;r\to\infty,\;\theta\in\{-\pi/6,-5\pi/6\},
\end{eqnarray}
where $m\geq 1$ is an integer.
The first family was considered in \cite{BBMSS} and
\cite{DF1,DF2}. We prove that the spectral
locus in the real $(a,c,\lambda)$-space $\R^3$ consists
of infinitely many smooth
analytic surfaces $S_n,\; n\geq 0$, each homeomorphic to
a punctured disc, and that an eigenfunction
corresponding to a point $(a,c,\lambda)\in S_n$ has
exactly $2n$ zeros which do not lie on
the imaginary axis. We study the shape and
position of these surfaces
by degenerating the quartic potential to the
previously studied $PT$-symmetric
cubic oscillator. 

The second quartic family (\ref{BS1}-\ref{BS2})
was introduced by Bender and Boettcher \cite{BB}.
It is quasi-exactly solvable (QES) in the sense
that for every integer $m\geq 1$ in the potential,
there are $m$ ``elementary'' eigenfunctions, each
having $m-1$ zeros. The part $Z_m$
of the spectral locus corresponding to
these elementary eigenfunctions is a
smooth connected curve in $\C^2$ \cite{EG2,EGarxiv}. In the
end of Section~4
we study the
intersection of this curve with the real $(a,\lambda)$-plane.
Similarly to the case of the $PT$-symmetric cubic,
this intersection consists of smooth
analytic curves $\Gamma_{m,n}^*$, $n=0,\ldots, 
\lceil m/2\rceil,$
and for $(a,\lambda)\in \Gamma_{m,n}^*$ the
eigenfunction has exactly $2n$ zeros that do not
lie on the imaginary axis.
For $n\leq m/2$ the part of $\Gamma_{m,n}^*$ 
over some ray $a>a_m$ consists of disjoint graphs of
two functions, and we have the following ordering:
$(a,\lambda)\in\Gamma_{m,n}^*,\;
(a,\lambda')\in\Gamma_{m,n+1}^*$ and $a>a_m$ imply that
$\lambda'>\lambda$.
Moreover, the QES spectrum for $a>a_m$ consists of the
$m$ smallest real eigenvalues.

The results of Section 4 permit us to answer the
question of Hellerstein and Rossi stated above
for polynomial potentials of degree 4:
All equations (\ref{hr}) that possess a solution
with $2n$ non-real zeros are parametrized by our
curves $\Gamma_{m,n}^*$ if the total number
of zeros is $m-1$,
and by our surfaces $S_n$ if the total
number of zeros is infinite.
\vspace{.1in}

Notation and conventions.

1. What we call Stokes lines is called by some authors ``anti-Stokes lines''
and vice versa. We follow terminology of Evgrafov and Fedoryuk
\cite{EF,F}.

2. We prefer to replace $z$ by $iz$
in $PT$-symmetric problems. 
Then potentials become real, and the difference between
$PT$-symmetric and self-adjoint problems
is that in $PT$-symmetric problems
the complex conjugation
{\em interchanges} the two boundary conditions, while
in self-adjoint problems both boundary
conditions remain fixed
by the symmetry.
The main advantage for us in this change of the variable is
linguistic: we frequently refer to ``non-real'' zeros.
The expression ``non-real'' excludes $0$,
while the expression ``non-imaginary'' does not.

We thank 
Kwang-Cheul Shin for his useful remarks and
for sending us the text of his lecture \cite{S2} and
Per Alexandersson for the plots of
Stokes complexes he made for this paper.

\section{Perturbation of eigenvalues and eigenfunctions}

We begin with recalling some facts about boundary value problem
(\ref{problem}) with potential (\ref{P}).
The {\em separation rays} are defined by 
$$\re\left(\int_0^z\sqrt{\zeta^d}\ d  \zeta\right)=0,
\quad\mbox{that is}\quad z^{d+2}<0.$$
These rays divide the plane into $d+2$ open sectors $S_j$ which we call 
{\em Stokes sectors}. We enumerate them by residues modulo $d+2$
counterclockwise. 

A solution $w$ of the differential equation (\ref{problem}) is called
{\em subdominant} in $S_j$ if $w(rz)\to 0$ as $r\to+\infty$, for all $z\in S_j$.
For every $j$, the space of solutions of the equation in (\ref{problem})
which are subdominant in $S_j$ is one-dimensional.
If $S_j$ and $S_k$ are {\em adjacent}, that is $j=k\pm 1\, \mod(d+2)$
then the corresponding subdominant solutions are linearly independent.

Let $S_j$ and $S_k$ be two non-adjacent Stokes sectors.
We consider the boundary conditions
\begin{equation}\label{bbc}
w\;\mbox{is subdominant in}\; S_j\;\mbox{and}\; S_k.
\end{equation}
Such boundary value problem has an infinite set of eigenvalues
tending to infinity. All eigenspaces are one-dimensional.
These facts were proved by Sibuya \cite{Sibuya} whose main tool were
special solutions normalized on one ray, which we call {\em Sibuya solutions}.
Precise definition is given below.
Our first goal is to prove continuous dependence of Sibuya solutions on
parameters.

We consider a family of polynomial potentials with parameters $(t,\a)$:
\begin{equation}\label{Q}
Q(z,t,\a)=tz^d+\sum_{j=0}^m a_jz^j,\quad m<d,\quad\a=(a_0,\ldots,a_m).
\end{equation}
Let $K\subset\C^{m+1}$ be a compact set which has a fundamental system
of open simply connected neighborhoods, and
such that $a_m\neq 0$ for $\a\in K$.
This compact $K$ will be fixed in all our arguments,
so our notation does not reflect dependence of the quantities introduced
below on $K$.
Let $L=\{ te^{i\theta}:t\geq t_0\}$ be a ray in $\C$.

Suppose that for some $\delta>0$ and for all $(t,\a)\in [0,1]\times K$
the following conditions hold:
\vspace{.1in}

{\bf a)} There exists $R>0$ such that
$|\arg z-\theta|\geq\delta$ for
all zeros $z$ of $Q(z,t,\a)$ such that $|z|>R$.

{\bf b)} Whenever $L$ intersects a vertical trajectory of
$Q(z,t,\a)dz^2$ at a point $z$, $|z|>R$,
the smallest angle between this trajectory and
$L$ at the intersection point is at least $\delta$.

{\bf c)} $\theta$ is not a separating direction of $Q(z,t,\a)dz^2$.
\vspace{.1in}

One can easily show that b) implies c). Condition c) simply means that
$L$ is neither a Stokes line for $tz^dsz^2$, nor
a Stokes line for $a_mz^mdz^2$.

Condition a) implies that there is a branch $Q_L^{1/4}$ of
$Q^{1/4}$ analytic on $[0,1]\times K\times L$.
Let $Q_L^{1/2}=(Q_L^{1/4})^2$ be the corresponding branch of $Q^{1/2}$.
We choose the original branch $Q_L^{1/4}$ in such a way that
$$\re Q_L^{1/2}(z)dz\to+\infty,\quad z\to\infty,\quad z\in L.$$
This is possible in view of condition c).

Let us say that $y=y_L(z,t,\a)$ is a Sibuya solution of
$$-y^{\prime\prime}+Q(z,t,\a)y=0$$
if
$$y(z)\sim Q_L^{-1/4}(z)
\exp\left(-\int_{z_0}^zQ_L^{1/2}(\zeta)d\zeta\right),\quad z\to\infty,\quad
z\in L.$$
Here $z_0=t_0e^{i\theta}$. Notice that a change of $t_0$, results
is multiplying the Sibuya solution by a factor that depends only
on $t$ and $\a$ but not on $z$.

\begin{theorem}\label{thm1}
Under the conditions {\rm a), b), c)} above, there
exists a unique Sibuya solution. It is an analytic function
of $(z,t,\a)$ in a neighborhood of $\C\times(0,1]\times K$,
continuous on $K_1\times[0,1]\times K$ for
every compact $K_1\subset\C_z$.
\end{theorem}

\begin{proof} Let $$\phi(z)=\int_{z_0}^zQ_L^{1/2}(\zeta)d\zeta,$$
where the integral is taken along $L$
This is an analytic function which maps $L$ onto
a curve $\gamma$. Condition b) implies that $\gamma$ is a graph
of a function (intersects every vertical line at most once),
the slope of $\gamma$ is bounded,
and $\phi$ maps bijectively some neighborhood of
$L$ onto a neighborhood of $\gamma$.

Let $\psi$ be the inverse function to $\phi$.

Setting $u=(Q_L^{1/4}\circ\psi)y\circ\psi,$
we obtain the differential equation
$$u^{\prime\prime}=(1-g(\zeta))u,$$
where the primes stand for differentiation with respect to $\zeta$, and
$$g(\zeta)=\left(\frac{5}{16}\frac{{Q^\prime}^2}{Q^3}-
\frac{1}{4}\frac{Q^{\prime\prime}}{Q^2}\right)\circ\psi(\zeta).$$
This is equivalent to the integral equation
\begin{equation}\label{inteq}
u(\zeta)=e^{-\zeta}+\int_\zeta^\infty
\left(e^{\zeta-\eta}-e^{-\zeta+\eta}\right)g(\eta)u(\eta)d\eta,
\end{equation}
where the path of integration is the part of $\gamma$ from  $\zeta$
to $\infty$. The integral equation is solved by successive approximation.
We set $u(\zeta)=v(\zeta)e^\zeta$, and obtain
\begin{equation}
\label{inteq2}
v(\zeta)=1+\frac{1}{2}\int_\zeta^\infty
\left( e^{2(\zeta-\eta)}-1\right)g(\eta)v(\eta)d\eta,
\end{equation}
which we abbreviate as $v=1+F(v).$
Setting $v_0=0$ and $v_{n+1}=1+F(v_n)$,
we obtain
$$\| v_{n+1}-v_n\|_\infty\leq\| v_n-v_{n-1}\|\int_{\zeta}^\infty|g(t)|dt.$$
Here we used the fact that $\Re(\zeta-\eta)<0$ on the curve of integration
because $\gamma$ is a graph of a function.
Now, if $\zeta>0$ is large enough, we have
\begin{equation}\label{twa}
\int_\zeta^\infty|g(\eta)|d\eta<1/2
\end{equation}
for all values of parameters $t\in[0,t_0],\a\in K$.
We state this as a lemma:

\begin{lemma} There exists $b\in L$ such that
for the piece $L_b$ of $L$ from $b$
to $\infty$ and for all $(t,\a)\in[0,1]\times K$
we have
$$\left|\int_{L_b}\left(
\frac{5}{16}\left(\frac{Q^\prime}{Q}\right)^2-
\frac{1}{4}\frac{Q^{\prime\prime}}{Q}\right)
Q^{-1/2}dz\right|<1/2.$$
\end{lemma}

The integral in this lemma equals to the integral in
(\ref{twa}) by the change of the variable $z=\psi(\zeta),
\sqrt{Q}dz=d\zeta$.
\vspace{.1in}

{\em Proof}.
Let $z_1,\ldots,z_d$ be all zeros of $Q_t$ listed with
multiplicity, in an order of
non-decreasing moduli. Suppose that $z_1,\ldots,z_M$
are in the disc $|z|<R$ while the rest are outside. Here $R$
is the number from condition a). We have
$$\frac{Q^\prime}{Q}=\sum_{k=1}^M\frac{1}{z-z_k}+\sum_{k=M+1}^d\frac{1}{z-z_k}=
\sigma_1(z)+\sigma_2(z),$$
and
$$\frac{Q^{\prime\prime}}{Q}=\left(\frac{Q^\prime}{Q}\right)^\prime
+\left(\frac{Q_t^\prime}{Q_t}\right)^2=\sigma_1^\prime(z)+\sigma_2^\prime(z)+\sigma_1^2(z)+2\sigma_1(z)\sigma_2(z)+\sigma_2^2(z).$$
First we estimate
\begin{equation}
\label{sigma1}
|\sigma_1|\leq \frac{M}{|z|-R},\quad
|\sigma_1^\prime(z)|\leq\frac{M}{(|z|-R)^2}.
\end{equation}
To estimate $\sigma_2$ we first use
condition a) to conclude that
\begin{equation}\label{c}
|z-z_k|\geq C_0|z|,\quad z\in L, \quad  k\geq M
\end{equation}
where $C_0$ depends only on $\delta$. Then
\begin{equation}
\label{sigma2}
|\sigma_2(z)|\leq C_1/|z|,\quad|\sigma_2^\prime|\leq C_2/|z|^2.
\end{equation}
Applying these inequalities, we obtain that
$$
\left|\frac{5}{16}
\frac{{Q_t^{\prime}}^2}{Q_t^2}-
\frac{1}{4}
\frac{Q_t^{\prime\prime}}{Q_t}
\right|\leq
\frac{C}{|z|^2},$$
on $L$, where $C$ depends only on $\delta$.

Now we write
$$|Q_t(z)|=t\prod_{j=1}^M|z-z_j|\prod_{j=M+1}^d|z-z_j|\geq
tC_3^d(|z|-R|)^M\prod_{j=M+1}^d|z_j|.$$
Here we used inequality
(\ref{c}) with interchanged
$z$ and $z_k$.
It is easy to see by Vieta's theorem that
$$\prod_{j=M+1}^d|z_j|\geq C_4t^{-1},$$
where $C_4$ depends only on $K$.
This shows that $|Q_t(z)|\geq C_6|z|^M.$

Now we use the fact that $L_b$ is a part of a ray from
the origin, so $|dz|=d|z|$ on $L_b$
Putting all this together we conclude that our integral
is majorized by the integral $$\int_{|b|}^\infty |z|^{-2-M/2}d|z|$$
which proves the lemma.

\vspace{.1in}

So the series $\sum v_n$ is convergent uniformly
in $\re\zeta>A$ for some $A>0$ and this convergence is uniform
with respect to $t,\a.$
Then an application of the theorem on
uniqueness and continuous dependence of initial
conditions for linear differential equations
shows that this convergence is
uniform also on compacts in the right half-plane
of the $\zeta$-plane.
 \end{proof}
\vspace{.1in}

Let $Z_t$ be a family of discrete subsets of the complex plane,
depending on a parameter $t$. We say that $Z_t$ {\em depends continuously}
on $t$ if there exists a family of entire functions $f_t\neq 0$
such that $Z_t$ is the set of zeros of $f_t$, and $f_t$ depends
continuously on $t$. Here the topology on the set of entire functions
is the usual topology of uniform convergence on compact subset of
the complex plane.

Consider the eigenvalue problem
\begin{equation}\label{eig}
-y^{\prime\prime}+Qy=\lambda y,\quad y(z)\to 0,\quad z\to\infty\quad z\in
L_1\cup L_2,
\end{equation}
where $Q$ is a polynomial in $z$ of the form (\ref{Q}),
and $L_1,L_2$ are two rays from the origin in the complex plane.
We say that a ray $L$ is {\em admissible} if conditions a), b) and c)
in the beginning of this section are satisfied. The notion of
admissibility depends on the parameter region $K$ participating
in conditions a), b) and c).

\begin{theorem}\label{thm2} If both rays $L_1$ and $L_2$ are admissible then
the spectrum of problem (\ref{eig}) is continuous for $(t,\a)\in [0,1]\times K$.
\end{theorem}

\begin{proof} Let $y_1$ and $y_2$ be the Sibuya solutions corresponding
to the rays $L_1$ and $L_2$. Then their Wronski determinant
$W=y_1^\prime y_2-y_1y_2^\prime$, where the primes indicate differentiation
with respect to $z$, is the spectral determinant of the problem (\ref{eig}),
and $W$ depends continuously on $(\a,\lambda)$ in view of Theorem~\ref{thm1}.
\end{proof}

The limit problem (\ref{eig}) for $t=0$ may have no eigenvalues,
this is the case when the rays $L_1$ and $L_2$ belong to adjacent sectors
of $Q(z,0,\a)$. In this case,
Theorem~\ref{thm2} says that the eigenvalues escape
to infinity as $t\to 0$.

\begin{theorem}\label{thm3} Consider the problem (\ref{eig}),
and suppose that the rays $L_1$ and $L_2$ are admissible, $\lambda(t,\a)$ is
an eigenvalue that depends continuously on $(t,\a)$ and has a finite
limit $\lambda(0,\a)$ as $t\to 0$.
Then there exists an eigenfunction $y(z,t,\a)$ of this problem,
corresponding to this eigenvalue, that depends continuously on $(t,\a)$.
\end{theorem}

\begin{proof}
Let $y_1(z,t,\a)$ be the Sibuya solution of (\ref{eig})
with $\lambda=\lambda(t,\a)$,
corresponding to the ray $L_1$.
It depends continuously on $(t,\a)$
by Theorem \ref{thm3}. Let $y_2(z,t,\a)$ be the Sibuya solution
of the same equation corresponding to $L_2$.
The assumption that $\lambda(z,t,\a))$ is an eigenvalue implies
that $y_1$ and $y_2$ are proportional.
This implies that $y_1$ tends to zero as $z\to\infty$ on $L_2$, so
$y_1$ satisfies both boundary conditions. Thus
$y_t$ is an eigenfunction that depends continuously on $t$.
\end{proof}
 
\section{Admissible rays}

In this section we give a criterion for a ray
to be admissible (see the definition before Theorem~\ref{thm2}).
We reduce the problem to the case of a binomial
$Q(z)=tz^d+cz^m,\;t\ne 0,\; c\ne 0,\; 0<m<d.$ 

We begin with recalling terminology.
Let $Q(z)$ be a polynomial, $z\in\C$.
A {\sl vertical line} of $Q(z)\,dz^2$ is an integral curve
of the direction field 
$Q(z)\,dz^2<0$. 
A {\sl Stokes line} of $Q$ is
a vertical line with one or both ends in 
the set of {\sl turning points} $\{z: Q(z)=0\}$.
The {\sl Stokes complex} of $Q$ is the union of the Stokes lines
and turning points. Examples of Stokes complexes are shown
in Figs.~\ref{fig1}-\ref{fig3}.

\begin{figure}[ht]
\centering
\subfigure[$z^4+iz^3$]{\label{fig1a}\includegraphics[scale=0.5]{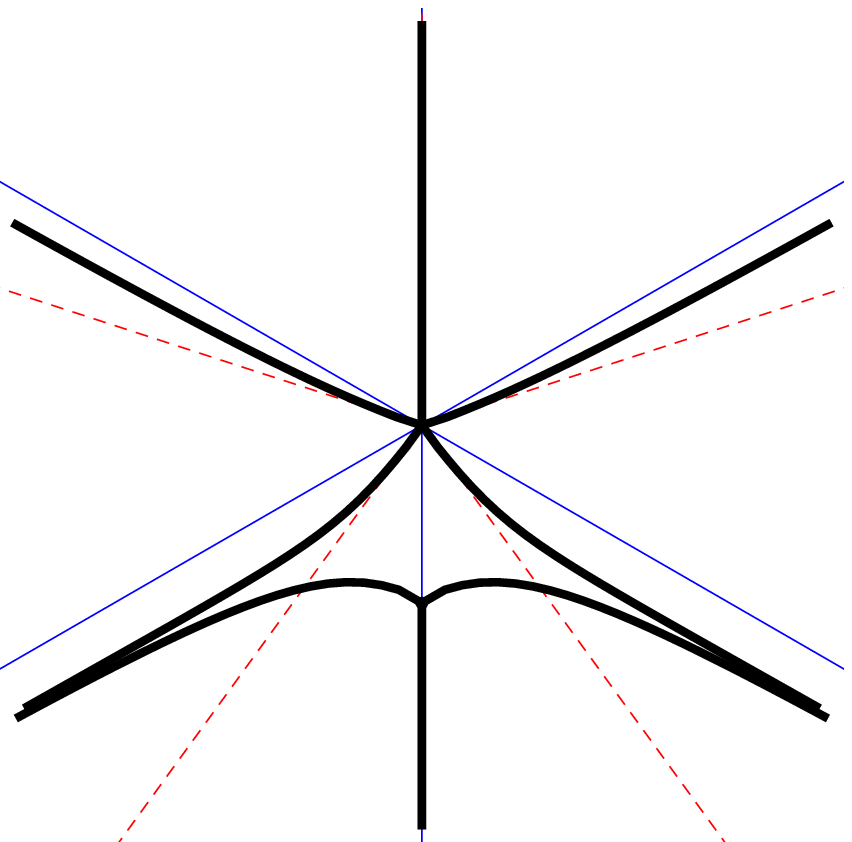}}
\hskip.2in\subfigure[$z^4+e^{\pi i/4} z^3$]{\label{fig1b}\includegraphics[scale=0.5]{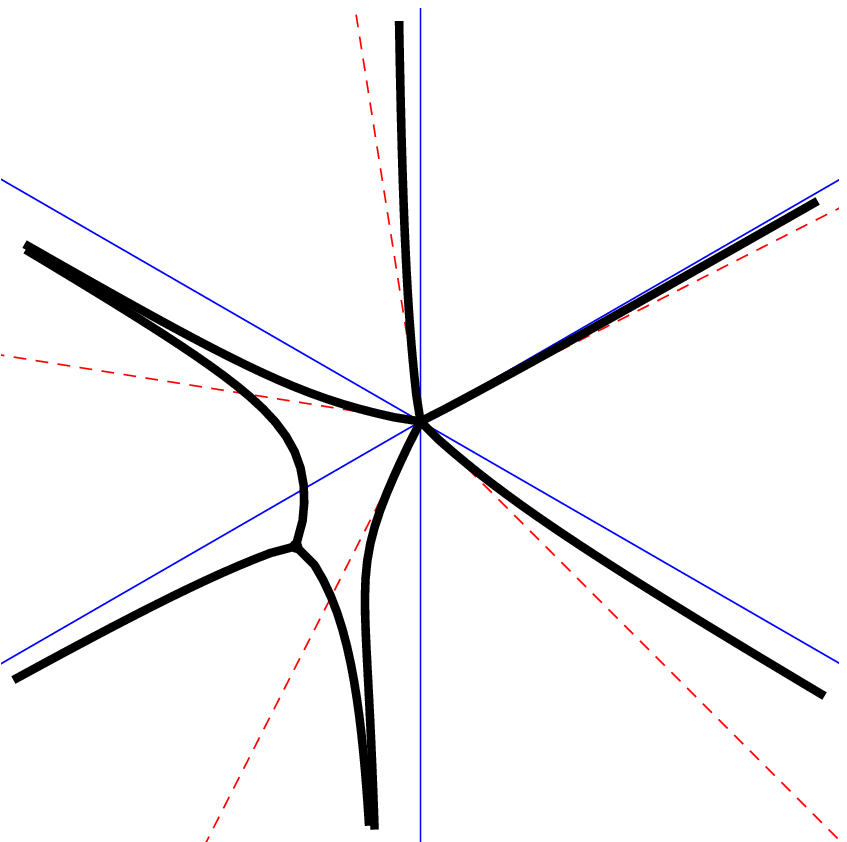}}
\caption{Stokes complexes of $z^4+iz^3$ and $z^4+e^{\pi i/4} z^3$.}\label{fig1}
\end{figure}

\begin{figure}[ht]
\centering
\subfigure[$z^4+z^3$]{\label{fig2a}\includegraphics[scale=0.5]{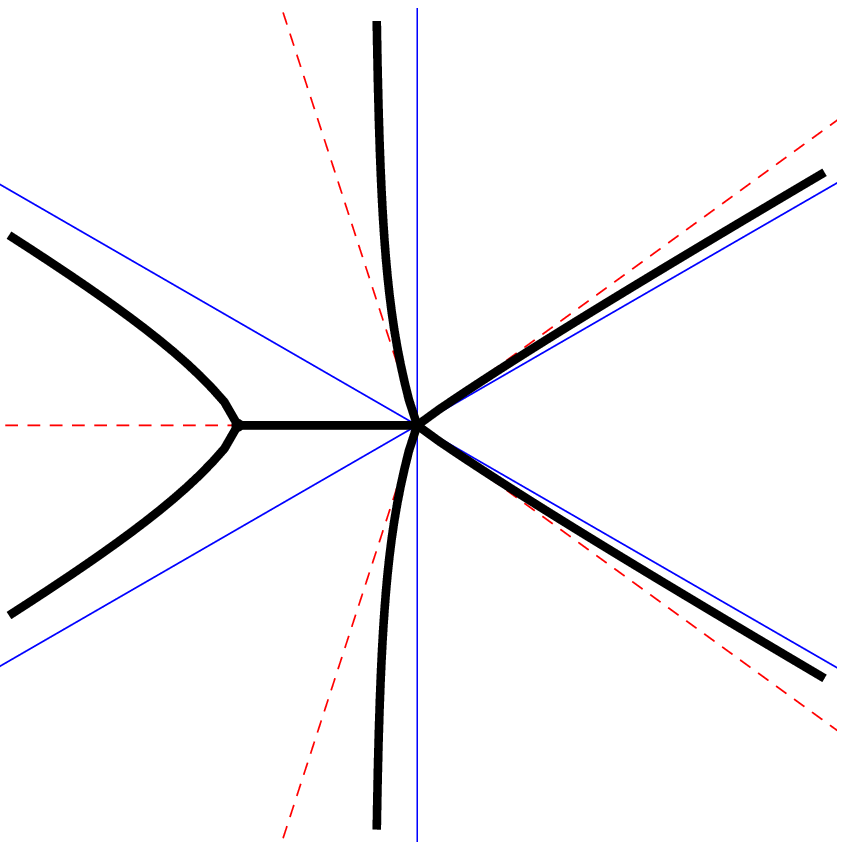}}
\hskip.2in\subfigure[$z^3-z^2$]{\label{fig2b}\includegraphics[scale=0.5]{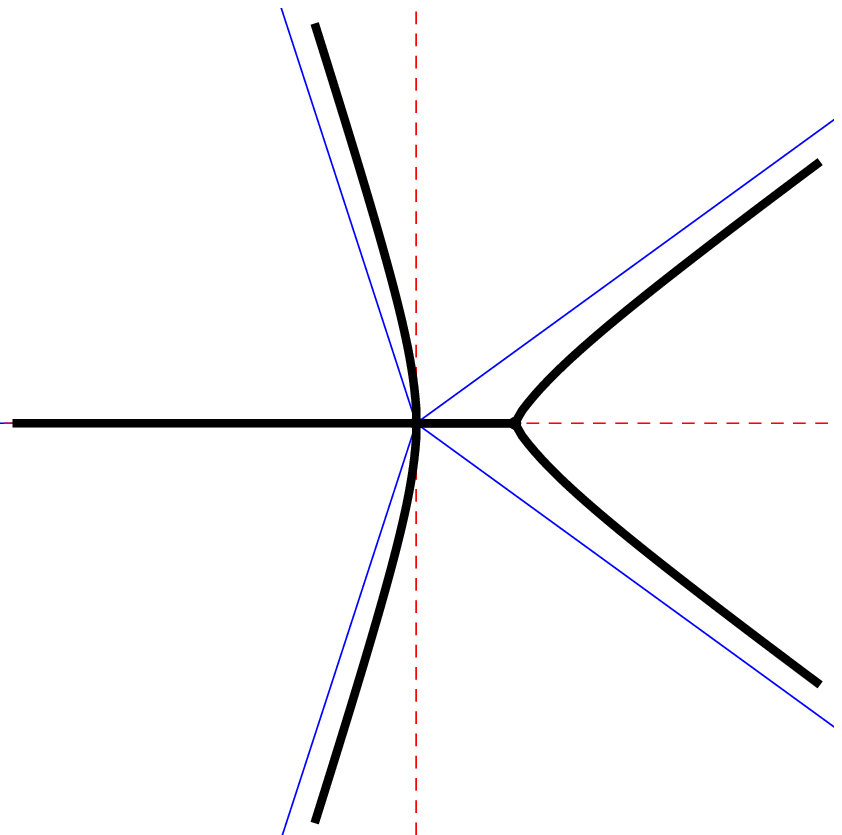}}
\caption{Stokes complexes of $z^4+z^3$ and $z^3-z^2$.}\label{fig2}
\end{figure}

\begin{figure}[ht]
\centering
\subfigure[$z^3+z^2$]{\label{fig3a}\includegraphics[scale=0.5]{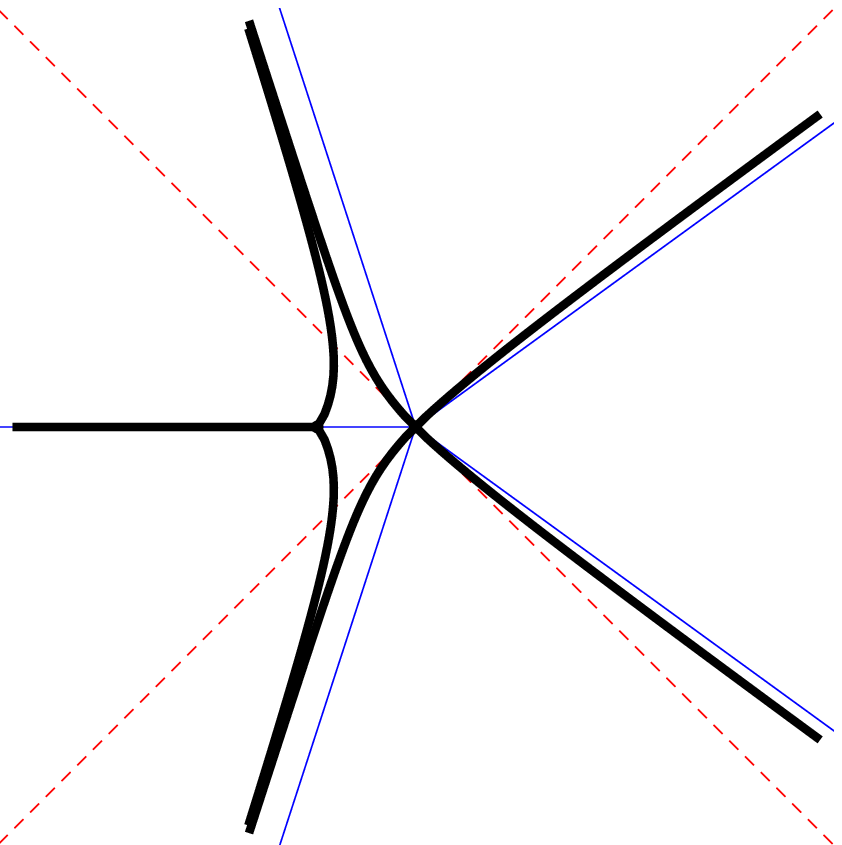}}
\hskip.2in\subfigure[$z^3+iz^2$]{\label{fig3b}\includegraphics[scale=0.5]{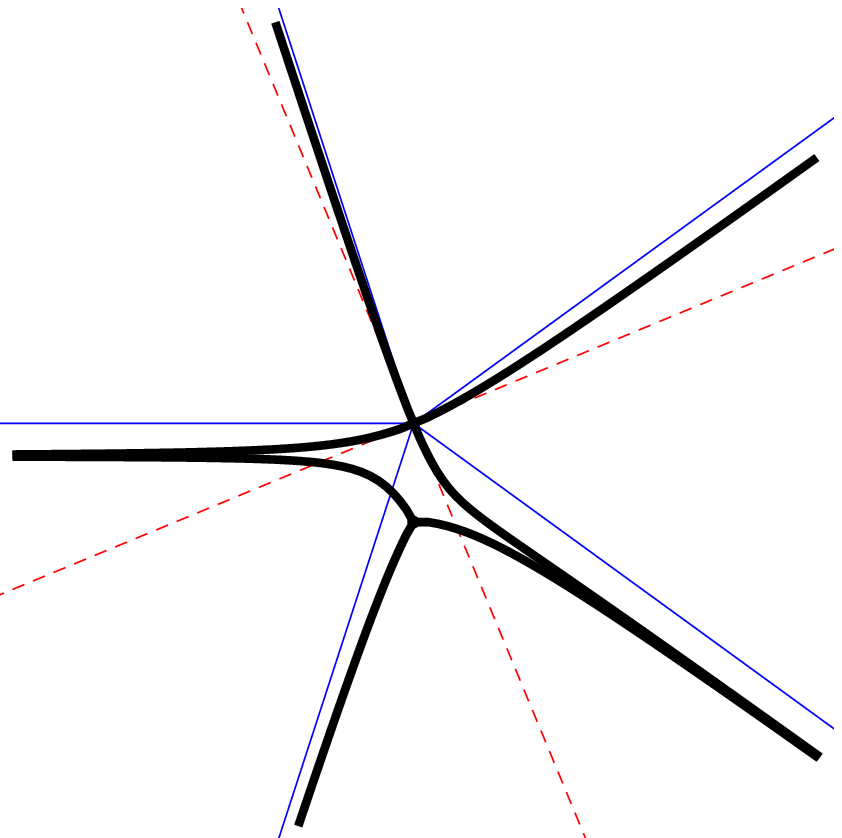}}
\caption{Stokes complexes of $z^3+z^2$ and $z^3+iz^2$.}\label{fig3}
\end{figure}

A {\sl horizontal line} of $Q$ is a vertical line of $-Q$.
Vertical and horizontal lines intersect orthogonally.
An {\sl anti-Stokes line} of $Q$ is a Stokes line of $-Q$.

Every Stokes line has one end at a turning point and
the other end either at a different turning point or
at infinity.
If $Q$ has a zero at $z_0$ of multiplicity $m$ then there
are $m+2$ Stokes lines with the endpoint at $z_0$;
they partition a neighborhood of $z_0$ into sectors of equal 
opening $2\pi/(m+2)$.
The $m+2$ anti-Stokes lines having one end at $z_0$ bisect these
sectors.

Let $L(\alpha)=\{z\in\C\setminus\{0\}:\mathrm{arg}\,z=\alpha\}$
and, for $0<\beta-\alpha<2\pi$,
$S(\alpha,\beta)=\{z\in\C\setminus\{0\}:\alpha<\mathrm{arg}\,z<\beta\}$.
For $R\ge 0$, let $D(R)=\{z\in\C:\abs{z}\le R\}$.
For a ray $L$ or a sector $S$, define
$L_R=L\setminus D(R),\; S_R=S\setminus D(R)$.
For a set $S\subset\C$, $\overline S$ is its closure in $\C$
and $\partial S=\overline S\setminus S$.

\begin{definition}\label{goodsector}
Let $Q=P_d+P_m$ be a binomial, where $P_d(z)=t z^d$ and $P_m(z)=c z^m$
are two non-zero monomials, $0<m<d$.
Let ${\mathcal S}(Q)$ be the partition of $\C\setminus\{0\}$ into
open sectors and rays defined by the Stokes lines
of the two monomials $P_d$ and $P_m$.
Let ${\mathcal R}(Q)$ be the refinement of ${\mathcal S}(Q)$
defined by the rays from the origin through the non-zero
turning points of $Q$.

A ray $L(\theta)$ is called {\em good} for $Q$ if it
is not one of the rays of  ${\mathcal R}(Q)$
and is not tangent to any vertical line of $Q$.
The last condition is equivalent to $L(\theta)\cap Z=\emptyset$ where
$Z=\{z:z^2 Q(z)\le 0\}$.
A sector $S$ of ${\mathcal R}(Q)$ is {\em good} for $Q$
if each ray $L(\theta)\subset S$ is good.
This is equivalent to $S\cap Z=\emptyset$.
\end{definition}

\begin{theorem}\label{allgood}
Let $Q=P_d+P_m$ be as in Definition \ref{goodsector}.
Then any sector of ${\mathcal R}(Q)$
containing an anti-Stokes line of $P_d$ is good for $Q$,
and any good ray belongs to one of such sectors.
\end{theorem}

\begin{proof}
Let $\R_+$ and $\R_-$ be the positive and negative
real rays (not including $0$).
Definition \ref{goodsector} implies that
a ray $L=L(\theta)$ is good if and only if the cone
$C_L=\{\alpha z^2 P_d(z)+\beta z^2 P_m (z),\; z\in L,\;
\alpha\in\R_+,\;\beta\in\R_+\}$ does not contain $\R_-$.

An anti-Stokes line of $P_d$ is a good ray unless
it is also a Stokes line of $P_m$, because $z^2 P_d(z)$
is real positive on it, hence $z^2 P_m(z)$ must be real
negative to make the sum real negative.

Suppose that $L$ is an anti-Stokes line of $P_d$,
and that $z^2 P_m(z)$ is either real positive or
belongs to the upper half-plane for $z\in L$.
Then either $C_L=\R_+$ or $C_L$ belongs to the upper half-plane.
When $L$ is rotated counterclockwise, the arguments
of the two monomials $z^2 P_d$ and $z^2 P_m$
restricted to $L$ are increasing.
Hence $C_L$ remains in the upper half-plane until
at least one of the monomials
becomes real negative on $L$, i.e., until $L$ becomes
a Stokes line of either $P_d$ or $P_m$, or both.
When $L$ is rotated clockwise, the arguments
of the two monomials restricted to $L$ are decreasing.
Since the argument of $P_d$ decreases faster than the argument
of $P_m$, the cone $C_L$ does not contain
negative real numbers until either $z^2 P_d$ on $L$
becomes real negative
or $\arg P_m-\arg P_d$ passes $\pi$, i.e., until $L$
either becomes a Stokes line of $P_d$ or passes a non-zero
turning point of $Q$.

The case when $P_m(z)$ belongs to the lower half-plane
on an anti-Stokes line of $P_d$ is done similarly.

Conversely, let $L$ be a ray which is not one of the rays
of ${\mathcal R}(Q)$
and such that $C_L$ does not contain negative real numbers.
Then either $L$ itself is an anti-Stokes line of $P_d$, or it can
be rotated to the
closest anti-Stokes line of $P_d$ preserving this property.

For example, if $z^2 P_d(z)$ belongs to the upper half-plane
for $z\in L$, then $\arg (z^2 P_d(z))-\pi<\arg(z^2 P_m(z))<\pi$ for $z\in L$.
Otherwise, either $L$ would be a Stokes line of $P_m$
(if $\arg(z^2 P_m(z))=\pi$),
or it would contain a turning point of $Q$
(if $\arg (z^2 P_d(z))-\pi=\arg(z^2 P_m(z))$),
or $C_L$ would contain $\R_-$.
When $L$ is rotated clockwise, since $\arg P_d$ decreases faster
than $\arg P_m$, $P_d$ would become real positive on $L$
before either $\arg P_m-\arg P_d$ becomes $\pi$ or
$z^2 P_m$ becomes real negative.

The case when $z^2 P_d(z)$ belongs to the lower half-plane on $L$
is done similarly, rotating $L$ counterclockwise.
\end{proof}

\begin{corollary}\label{subsectors}
Let $S=S(\alpha,\beta)$ be a Stokes sector of $P_d$ such that
the anti-Stokes line $L\subset S$ of $P_d$ does not contain a turning point
of $Q$. Then $S$ contains a good subsector.
\end{corollary}

\begin{proof}
Since $L$ does not contain a turning point of $Q$, it cannot be a
Stokes line of $P_m$. Hence $L$ belongs to a sector of 
${\mathcal R}(Q)$, which is good by Theorem \ref{allgood}.
\end{proof}

\begin{theorem}\label{hexagon}
Let $Q=P_d+P_m$ be as in Definition~\ref{goodsector}.
Let $H=\{(x,y)\in\R^2: |x|<\pi,\,|y|<\pi,\,|y-x|<\pi\}$ be a hexagon in $\R^2$.
Then $\R_+$ is a good ray for $Q$ if and only if
$(\arg t,\,\arg c)\in H$.
Here the values of $\arg t$ and $\arg c$ are taken in $(-\pi,\pi]$.

A ray $L(\theta)$ is a good ray for $Q$ if and only if
$(\arg t,\,\arg c)$ belongs to $H$ translated by
$(2\pi k-(d+2)\theta,\;2\pi l-(m+2)\theta)$ for some integers $k$ and $l$.
\end{theorem}

\begin{proof}
For $t>0$ and $c>0$, $\R_+$ is an anti-Stokes line of both $P_d$
and $P_m$, hence it is a good ray.
It is not a Stokes line of either $P_d$ or $P_m$ when
$|\arg t|<\pi$ and $|\arg c|<\pi$.
It does not contain a turning point if $|\arg t-\arg c|<\pi$.

For $|\arg t|<\pi$, the anti-Stokes line $L$ of $P_d$ closest to $\R_+$
has the argument $-\arg t/(d+2)$.
It is a Stokes line of $P_m$ if $\arg c-(m+2)\arg t/(d+2)=\pi(2k+1)$
for an integer $k$.
Since $0<m<d$, the lines $y-(m+2)x/(d+2)=\pi(2k+1)$ do not intersect $H$.
Hence $L$ does not contain a turning point of $Q$ when $(\arg t,\,\arg c)\in H$.
This implies that the closure of the sector $S$ bounded by $\R_+$ and $L$
does not contain Stokes lines of either $P_d$ or $P_m$,
and does not contain non-zero turning points of $Q$,
for all $(t,c)$ such that  $(\arg t,\,\arg c)\in H$.
Due to Theorem \ref{allgood}, $\R_+$ is a good
ray for $Q$ with these values of $t$ and $c$.

If one of the three inequalities defining $H$ becomes an equality,
$\R_+$ becomes either a Stokes line of one of the two monomials of $Q$,
or contains a turning point of $Q$.

When $(\arg t,\,\arg c)$ crosses one of the two segments
$|\arg t|<\pi$, $|\arg c|<\pi,$ $|\arg t-\arg c|=\pi$
of the boundary of $H$, a turning point of $Q$ crosses
$\R_+$ and remains inside the sector $S$ for all values
$(t,c)$ such that $|\arg t|<\pi,\;|\arg c|<\pi,\;
|\arg c-(m+2)\arg t/(d+2)|<\pi$.
Due to Theorem \ref{allgood}, $\R_+$ is not a good
ray for $Q$ with these values of $t$ and $c$.

This implies that $\R_+$ is not a good ray for $Q$
when $(\arg t,\,\arg c)\notin H$.

The statement for $L(\theta)$ is reduced to the statement
for $\R_+$ by the change of variable $z=u e^{i\theta}$
in the quadratic differential $Q(z)\,dz^2$.
\end{proof}

\begin{example}\label{4cases}
Let us investigate when the two rays of the real axis
are good for a binomial $Q(z)=t z^d+c z^m$.
According to Theorem \ref{hexagon}, $\R_+$ is a good ray
when $(\arg t,\,\arg c)\in H$ (green hexagon in Fig.~\ref{fig4}a).

\bigskip
\begin{figure}[ht]
\centering
\includegraphics[scale=0.6]{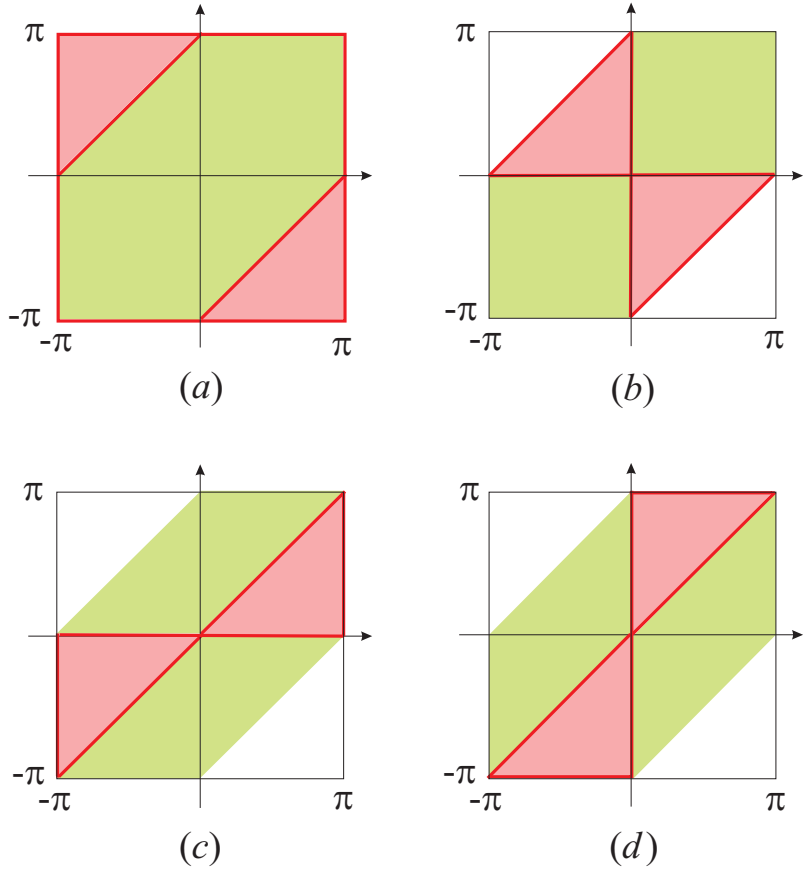}
\caption{The values of $(\arg t,\,\arg c)$ in Example \ref{4cases}
where both $\R_+$ and $\R_-$ are good (green)
and where $\R_-$ is not good (red). (a) $d$ and $m$ even; (b) $d$ and $m$ odd;
(c) $d$ even, $m$ odd; (d) $d$ odd, $m$ even.}\label{fig4}
\end{figure}

For $\R_-$ ($\theta=\pi$ in Theorem \ref{hexagon}) there are 4 cases,
depending on the parity of $d$ and $m$.

a) If $d$ and $m$ are even, both $\R_+$ and $\R_-$ are good rays
for any $Q$ such that $(\arg t,\arg c)\in H$.

b) If $d$ and $m$ are odd, $\R_-$ is a good ray for $Q$ when
$(\arg t,\arg c)$ belongs to the
complement in $(-\pi,\pi]^2$ of the two triangles 
(red area in Fig.~\ref{fig4}b) with the vertices
$(0,0)$, $(0,\pi),$ $(-\pi,0)$ and
$(0,0),$ $(\pi,0),$ $(0,-\pi)$,
respectively.  Both $\R_+$ and $\R_-$ are good rays for $Q$ when
$(\arg(t),\arg(c))$ belongs to the union of the two squares
$(0,1)^2$ and $(-1,0)^2$ (green area in Fig.~\ref{fig4}b).

c) If $d$ is even and $m$ is odd, $\R_-$ is a good ray for $Q$ when
$(\arg t,\arg c)$ belongs to the
complement in $(-\pi,\pi]$ of the two triangles 
(red area in Fig.~\ref{fig4}c) with the vertices
$(0,0),$ $(\pi,0),$ $(\pi,\pi)$ and $(0,0),$ $(-\pi,0),$
$(-\pi,-\pi)$,
respectively.  Both $R_+$ and $R_-$ are good rays for $Q$ when
$(\arg(t),\arg(c))$ belongs to the union of the two open
parallelograms (green area in Fig.~\ref{fig4}c) with the generators
$(\pi,0)$, $(-\pi,-\pi)$ and $(\pi,\pi),$ $(-\pi,0)$, respectively.

d) If $d$ is odd and $m$ is even, $\R_-$ is a good ray for $Q$ when
$(\arg t,\arg c)$ belongs to the
complement in $(-\pi,\pi]^2$ of the two triangles 
(red area in Fig.~\ref{fig4}d) with the vertices
$(0,0),$ $(0,\pi),$ $(\pi,\pi)$ and $(0,0),$ $(0,-\pi),$
$(-\pi,-\pi)$,
respectively.  Both $R_+$ and $R_-$ are good rays for $Q$ when
$(\arg(t),\arg(c))$ in the union of the two open parallelograms
(green area in Fig.~\ref{fig4}d) with the generators
$(0,\pi),\,(-\pi,-\pi)$ and $(0,-\pi),\,(\pi,\pi)$, respectively.

In particular, if $t$ is on the positive imaginary axis,
both  $R_+$ and $R_-$ are good rays for $Q$ when
$-\pi/2<\arg(c)<\pi$ in
case (a), $0<\arg(c)<\pi$ in case (b), $-\pi/2<\arg(c)<0$ or
$\pi/2<\arg(c)<\pi$ in case (c), $-\pi/2<\arg(c)<\pi/2$ in case (d).
\end{example}

By definition, a good ray $L$ is not tangent to the vertical lines of $Q$
and is not a Stokes line of either $P_d$ or $P_m$.
Since the angles between $L$ and vertical lines of $Q$ have
non-zero limits at the origin and at infinity, there is
a lower bound for these angles on $L$.
This lower bound depends continuously on $L$, hence there
is a common lower bound for these angles for all rays in
a proper subsector $T$ of a good sector $S$
(such that $\overline{T}\setminus\{0\}\subset S$).

The good sectors in Theorem \ref{allgood}
depend continuously on the arguments of the coefficients
$t$ and $c$ of monomials $P_d$ and $P_m$, except when
a good sector degenerates to a ray that is a Stokes line
of $P_m$ and an anti-Stokes line of $P_d$.

The lower bounds for the angles between a good ray $L$ and vertical
lines, and for the values of $R$, depend continuously on the
arguments of $t$ and $c$, except when a good sector containing
$L$ degenerates.

Now we show that a good ray for a monomial $tz^d+cz^m$
is admissible for the potential (\ref{Q}), with 
$$K=\{(a_0,\ldots,a_{m-1},c):\sup_{0\leq j\leq m-1}|a_j|\leq M\},$$
for every positive $M$.

\begin{lemma}\label{adm} Consider
the polynomials $Q(z)=tz^d+cz^m+q(z)$,
where $t\geq 0$, $|q(z)|\leq M|z|^{m-1}$ and $|c|=1$.
Let $S$ be a sector whose closure does not contain turning
points of $tz^d+cz^m$.

For every $\epsilon>0$ there exists $R>0$ depending on
$\epsilon,M,S$, such that for every $t\geq 0$:

(i) The set $S\cap\{ z:|z|\geq R\}$
does not contain turning points of $Q$, and

(ii) If $v(z)$ and $v'(z)$ are the vertical directions
of $Q(z)dz^2$ and $(tz^d+cz^m)dz^2$, respectively,
then $|\arg v(z)-\arg v'(z)|<\epsilon$ for $z\in S,\;|z|>R.$
\end{lemma}
\begin{proof}
We have $|tz^d+cz^m|\geq c|z|^m$ for $z\in S$,
so $|q(z)/(tz^d+cz^m)|\leq c^{-1}|z|^{-1}$, and (i), (ii) hold
when $R$ is large enough.
\end{proof}

\section{PT-symmetric potentials and
linear differential equations having 
solutions with prescribed number of non-real zeros}

Hellerstein  and Rossi asked the following
question \cite[Problem 2.71]{HC}.
Let
\begin{equation}\label{1-1}
w^{\prime\prime}+Pw=0
\end{equation}
be a linear differential equation with polynomial
coefficient $P$.
Characterize all polynomials $P$ such that
the differential equation
admits a solution with infinitely many zeros,
all of them real.

This problem was investigated in
\cite{T,G1,G2,HR,S1,EM}.
Recently K. Shin \cite{S2} announced
a description of polynomials $P$
of degree $3$ or $4$ such that
equation (\ref{1-1}) has a solution with infinitely many zeros,
{\em all but finitely many of them} real. It turns out that
equations (\ref{1-1}) with this property are
equivalent to (\ref{tr}) or (\ref{BB}) of the Introduction
by an affine change of the independent variable.

Here we use the methods of \cite{EG1,EG2} to
parametrize polynomials $P$ of degrees $3$ and $4$
such that equation (\ref{1-1})
has a solution with {\em prescribed}
number of non-real zeros.

We begin with degree $3$. 

\begin{theorem}\label{cubicth} For each integer $n\geq 0$
there exists a simple 
curve $\Gamma_n$ in the plane $\R^2$ which is the image of
a proper analytic embedding of the real line and which has
the following properties.

For every $(a,\lambda)\in\Gamma_n$ the equation
\begin{equation}
\label{4}
-w^{\prime\prime}+(z^3-az+\lambda)w=0
\end{equation}
has a solution $w$ with $2n$ non-real zeros. Real zeros 
belong to a ray $(-\infty,x_0)$ and there are
infinitely many of
them. This solution satisfies
$\lim_{t\to\pm\infty} w(it)=0$.

The union $\cup_{n=0}^\infty\Gamma_n$ coincides with
the real part of the spectral locus of (\ref{tr}).

The projection $(a,\lambda)\mapsto a$, $$\Gamma_n\cap\{
(a,\lambda):a\geq 0\}\to\{ a:a\geq 0\}$$ is a
$2$-to-$1$ covering map.
The curves $\Gamma_n$ are disjoint, and for $a\geq 0$ and $n\geq 0$,
if $(a,\lambda)\in
\Gamma_{n}$ and $(a,\lambda^{\prime})\in \Gamma_{n+1}$ then $\lambda<\lambda'$.
\end{theorem}

\begin{figure}[ht]
\centering
\includegraphics[scale=0.6]{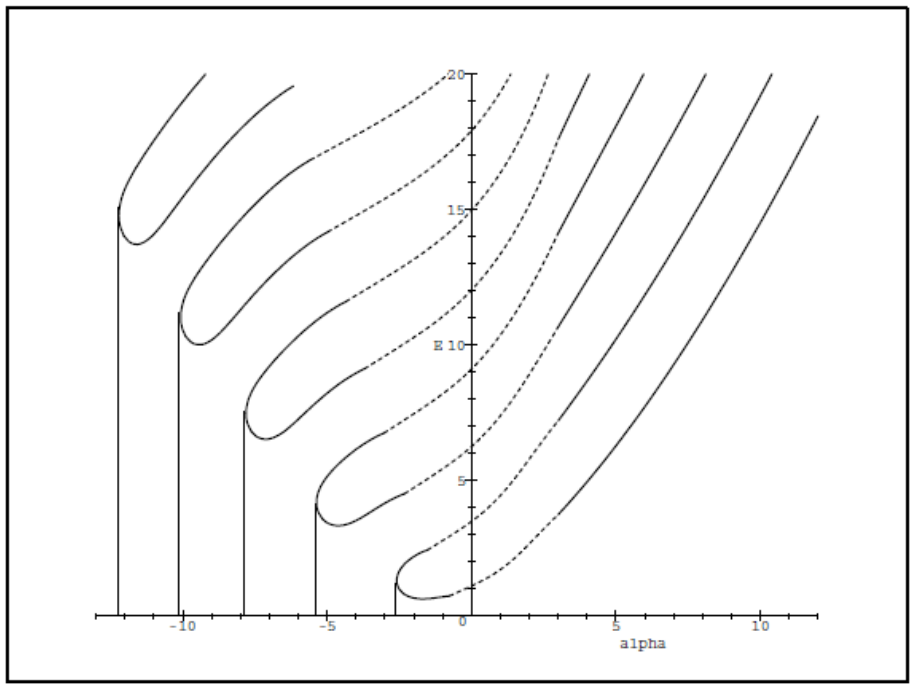}
\caption{Curves $\Gamma_n,\; n=0,\ldots,4$
in the $(a,\lambda)$ plane (Trinh, 2002).}\label{trinh-cubic}
\end{figure}

Equation (\ref{4}) is equivalent to
the PT-symmetric equation (\ref{tr})
in the
Introduction by the change of the independent
variable $z\mapsto iz$.
Computer experiments strongly suggest that the projection
$(a,\lambda)\mapsto a$ is $2$-to-$1$ on the whole curve
$\Gamma_n$ except one critical point of this projection,
and that the whole curve $\Gamma_{n+1}$ lies above
$\Gamma_{n}$.

Fig.~\ref{trinh-cubic}, taken from Trinh's thesis \cite{Trinh}
(see also \cite{DT}),
shows a computer
generated picture of the curves $\Gamma_n$.

As a corollary from Theorem~\ref{cubicth} we obtain that
every eigenvalue $\lambda_n(a),\; a\geq 0$ of (\ref{tr}),
when analytically continued to the left along the $a$-axis,
encounters a singularity for some $a<0$.
According to Theorem 2 of \cite{EG2} this singularity is an
algebraic ramification point.

\begin{proof}
Consider the Stokes sectors of equation (\ref{4}).
We enumerate them counter-clockwise as
$S_0,\ldots S_4$ where $S_0$ is bisected
by the positive real
axis. Consider the set $G$ of all
real meromorphic functions $f$
whose Schwarzian derivatives are real polynomials
of the form $-2z^3+a_2z^2+a_0$,
and whose asymptotic values in the sectors
$S_0,\ldots,S_4$ are $\infty,0,b,\overline{b},0$,
respectively, where $b=e^{i\beta},\beta\in (0,\pi)$.
Such functions are described by certain cell decompositions
of the plane \cite{EG2}. By a cell decomposition we understand
a representation of a space $X$ as a union of disjoint cells.
This union is locally finite, and the boundary of each cell
consists of cells of smaller dimension.
The $0$-cells are points, vertices of the decomposition.
The $1$-cells are embedded open intervals, 
the edges, and the $2$-cells are embedded open discs, 
faces of a decomposition. Two cell decompositions of a space $X$
are called equivalent if they correspond to each other  via
an orientation-preserving homeomorphism of $X$. 

To describe functions of the set $G$,
we begin with the cell
decomposition $\Phi$ of the Riemann sphere shown in Fig.~\ref{3loops}.

\begin{figure}[hb]
\centering
\includegraphics[scale=0.6]{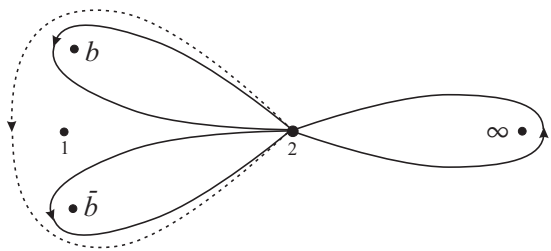}
\caption{Cell decomposition $\Phi$ of the image sphere (solid lines).}
\label{3loops}
\end{figure}

It consists of one vertex at $2$ and three edges which are
simple disjoint loops around $b,\overline{b}$ and $\infty$,
so that the loop around $\infty$ is symmetric with respect to
complex conjugation while the loops about $b$ and $\overline{b}$
are interchanged by the complex conjugation.
The point $0$ is outside the Fig.~\ref{3loops}.
The dotted line is not
discussed here; it is needed for the future. Also for the future use,
we assume that the loop around $\infty$ passes through the point $-2$,
and that this loop is symmetric with respect to
the reflection $z\mapsto-\overline{z}$.

So our cell decomposition has one vertex, three edges and four faces.
The faces are labeled by points $b,\overline{b},\infty$ and $0$
which are inside the faces. (So three faces are bounded by single
edge each, while one face (labeled with $0$) is bounded by three edges).

Suppose now that we have a local homeomorphism $g:\C\to\bC$
such that the restriction
\begin{equation}\label{cover}
g:\C\backslash g^{-1}(A)\to\bC\backslash A,
\end{equation}
where $A=\{ b,\overline{b},\infty,0\}$, is a covering map,
and $\Psi=g^{-1}(\Phi)$.
Then we preimage $\Psi=g^{-1}(\Phi)$ will be a cell decomposition
of the plane $\C$. 
Now suppose that a cell decomposition $\Psi$ of the plane is given
in advance,
and suppose that its local structure is the same as that
of $\Phi$. This means that the faces of $\Psi$ are labeled by
elements of the set $A$, and that a neighborhood of each vertex of $\Psi$
can be mapped onto a neighborhood of the vertex of $\Phi$ by an
orientation-preserving homeomorphism, respecting the labels
of the faces. Then there exists a local homeomorphism 
$g:\C\to\bC$ such that (\ref{cover}) is a covering map.
We use the following cell decomposition to construct $g$:

\begin{figure}[hb]
\centering
\includegraphics[scale=0.6]{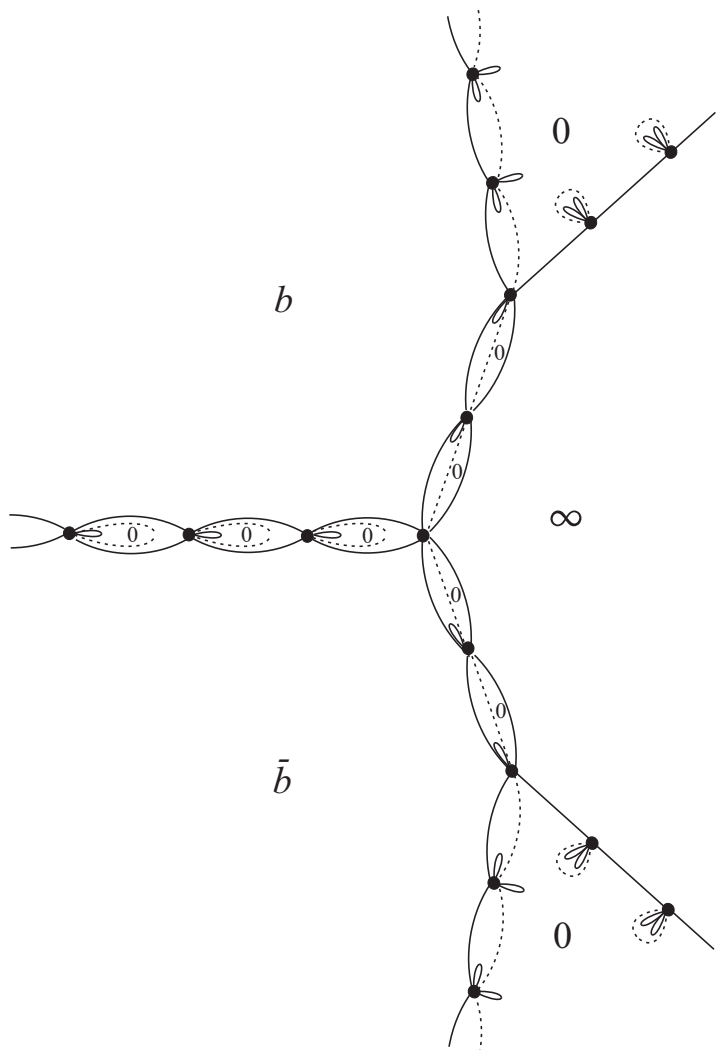}
\caption{Cell decomposition $\Psi_n$ for $n=2$ (solid lines).}
\label{cubic-r}
\end{figure}

The five ``ends'' extend to infinity periodically. This cell decomposition
depends on one integer parameter $n\geq 0$ which is the number of
$0$-labeled faces between the neighboring
``ramification points''.
Only some face labels are shown but the reader can easily recover
all other labels from the condition that in a neighborhood
of each vertex $\Psi_n$ is similar to a neighborhood of the vertex of $\Phi$.
The dotted lines are not a part of our cell
decomposition; they are preimages of the dotted line
in Fig.~19, and are added for a future need.
$\Psi_n$ is symmetric with respect to the real line, this permits
to make our local homeomorphism $g$ symmetric, that is
$g(\overline{z})=\overline{g(z)}$. 
This construction defines the map $g$ up to pre-composition with a
symmetric homeomorphism $\phi:\C\to\C$ of the domain of $g$.
A fundamental result of R. Nevanlinna ensures that this homeomorphism
$\phi$ can be chosen in such a way that $f=g\circ\phi$
is a meromorphic function which is real in the
sense that $f(\overline{z})=\overline{f(z)}$.
We refer to \cite{EG2} for the discussion of this construction
in our current context; in fact \cite{EG2} contains a simple alternative
proof
of Nevanlinna's theorem. Nevanlinna's original
proof is explained in modern language in \cite{E};
the original paper of Nevanlinna is
\cite{Nev}.

The meromorphic function $f$ is defined by the cell decomposition $\Psi_n$ and
parameter $b$ up to pre-composition with a real affine map $cz+d$.
Furthermore, the Nevanlinna theory says that the Schwarzian derivative
of $f$ is a polynomial of degree exactly $3$ (the number
of unbounded
faces of $\Psi_n$ minus $2$). We
pre-compose $f$ with
a real affine map to normalize 
this polynomial to have leading
coefficient $-2$ and 
zero coefficient at $z^2$. Thus
\begin{equation}\label{schwarz}
\frac{f^{\prime\prime\prime}}{f^\prime}-\frac{3}{2}\left(
\frac{f^{\prime\prime}}{f^\prime}\right)^2=-2(z^3-a z+\lambda).
\end{equation}
As $f$ is real, $a$ and $\lambda$ are also real.
Now $f$ is uniquely defined by the properties that it
satisfies a differential equation (\ref{schwarz}),
has asymptotic values $\infty,0,b,\overline{b},0$ in
the sectors $S_0,\ldots,S_4$, respectively, and
that $f^{-1}(\Phi)$ equivalent to $\Psi_n$
(Fig.~\ref{cubic-r} for $n=2$)
by an orientation-preserving
homeomorphism of the plane commuting
with the reflection $z\mapsto\overline{z}$.
The statement on asymptotic values implies that
$$f(it)\to 0, \quad t\to\pm\infty.$$
Furthermore, $f$ depends analytically on $b$, when $b$ is in the upper
half-plane, and thus we obtain a real analytic map $b\mapsto(a,\lambda)$.
This map is evidently invariant with respect to transformations
$b\mapsto tb,\; t\in\R\backslash\{0\}$, this is because the Schwarzian
derivative in the right hand side of (\ref{schwarz}) does not
change when $f$ is replaced by $tf$.

Thus for every $n$,
we have a one-parametric family $G_n\subset G$ of meromorphic
functions, 
parametrized by $\beta\in (0,\pi)$, $b=e^{i\beta}$.
Taking the Schwarzian derivative
we obtain a map $F_n: (0,\pi)\to\R^2$, 
$\beta\mapsto(a,\lambda)$.
This map is known to be a proper real analytic immersion
\cite{Bakken}.
It is easy to see that it is injective:
two solutions of the same Schwarz equation may differ
only by post-composition with a fractional-linear map,
and this fractional-linear map must be identity by our
normalization of asymptotic values.

For the same reasons the images of $F_n$ are disjoint:
for different $n$, our functions have (topologically)
different line complexes. The images of $F_n$ are
our curves $\Gamma_n$.

Now we prove that the union of $\Gamma_n$
coincides with the real part of the
spectral locus of (\ref{BB}).

Our functions $f\in G$ 
can be written
in the form $f=w/w_1$ where $w$ and $w_1$
are two linearly independent solutions
of equation (\ref{4}) with some real $a$ and $\lambda$. 
We can choose
$w$ and $w_1$ to be real entire functions. Condition that
$f(it)\to 0$ as $t\to\pm\infty$ implies that $w(it)\to 0$
for $t\to\pm\infty$ so $w$ is an eigenfunction
of the spectral problem
\begin{equation}
\label{shin}
-w^{\prime\prime}+(z^3-a z+\lambda)w=0,
\quad w(\pm i\infty)=0,
\end{equation}
which is equivalent to (\ref{tr}) by the change of
the independent variable $z\mapsto iz$.

So our curves $\Gamma_n$ belong to the real part
of the spectral locus of (\ref{shin}) or (\ref{tr}).

Now, let $\lambda$ be a real eigenvalue of the problem
(\ref{shin}), $w$ a corresponding eigenfunction.
Choose a point $x_0$ on the real line such that
$w(x_0)\neq 0$ and normalize $w$ so that $w(x_0)=1$.
Then 
$w^*(z)=\overline{w(\overline{z})}$ is an eigenfunction
with the same eigenvalue, so $w^*=cw$ for some constant
$c\neq 0$. Substituting $x_0$ gives that $c=1$.
So $w$ is real.

Let $w_1$ be a solution of the same equation as $w$ but 
satisfying $w(x)\to 0,\; x\to+\infty$. We normalize $w_1$
so that $w_1$ is real in the same way as we normalized
$w$. Then $f=w/w_1$ is a real meromorphic function
whose Schwarzian derivative is a cubic polynomial
with top coefficient $-2$, and the asymptotic values
in $S_j$ are $\infty, 0, b,\overline{b},0$.
We can change the normalization of $w_1$ multiplying it by
any real non-zero constant. In this way we
achieve that $b=e^{i\beta}$ for some $\beta\in (0,\pi).$
So $f$ belongs to the class $G$.

\begin{lemma}\label{GnG} $G=\cup_{n=0}^\infty G_n$.
\end{lemma}

\begin{proof} Let $f\in G$. Consider the cell decomposition
$X=f^{-1}(\Phi)$. We have to prove that $X=\Psi_n$
for some $n\geq 0$. To do this, we follow \cite{EG2}.
We first remove all loops from $X$, and then replace each
multiple edge by a single edge, and denote the resulting
cell decomposition by $Y$.
Notice that the cyclic order
$(\infty,b,\overline{b})$ in Fig.~\ref{3loops} is consistent
with the cyclic order $(\infty,0,b,\overline{b},0)$
of the Stokes sectors in the $z$-plane. 
By \cite[Proposition 6]{EG2}, this implies that
the $1$-skeleton of $Y$ is a tree. This infinite tree
is properly embedded in the plane, has $5$ faces,
is symmetric with respect to the real line, and has
two faces labeled with $0$ which are interchanged
by the symmetry. Moreover, the faces of $Y$ are in
one-to-one correspondence with the Stokes sectors,
and the face corresponding to $S_0$ is bisected
by the positive ray.
One can easily classify all trees with these
properties. They depend of one integer parameter $n\geq 0$
which is the distance between the ramification point
in the upper half-plane and the ramification point
on the real axis. Now we refer to 
\cite[Proposition 7]{EG2} that the tree $Y$ uniquely
defines the cell
decomposition $X$. This shows that $X=\Psi_n$ for some
$n\geq 0$.

Meromorphic function $f$ is defined by the cell
decomposition $X$ and the parameter $b$ up to an affine change of
the independent variable. Normalizing it as in
(\ref{schwarz}) gives $f\in G_n$.
\end{proof}

We conclude that the union of our curves $\Gamma_n$
in the right half-plane $a\geq 0$
coincides with the real part of the spectral
locus of (\ref{shin}).

Now we study the shape of the curves $\Gamma_n$.
The boundary value problem (\ref{shin})
was considered by Shin \cite{Shin}, Delabaere and Trinh
\cite{Trinh,DT}.
The spectrum of this problem is discrete,
simple and infinite.
It is known \cite{Shin} that for $a\geq 0$
all eigenvalues of
this problem are real and positive. It follows from this
result that there are real analytic
curves $\lambda=\gamma_k(a),\;
k=0,1,2,\ldots,$ such that for each $k$, $\gamma_k(a)$
is an eigenvalue of the problem (\ref{shin}), and
$\gamma_k<\gamma_{k+1},\; k=0,1,2,\ldots$.
So the part of the real spectral locus in $\{(a,\lambda):
a\geq 0\}$ is
the union of the graphs of $\gamma_k$.

Next we prove that the intersection of
$\Gamma_n$ with the half-plane $a\geq 0$ consists
of $\gamma_{2n}$ and $\gamma_{2n+1}$.
For this purpose we study what happens
to eigenvalues and eigenfunctions of the
problem (\ref{shin}) as $a\to+\infty$.

A different approach to the asymptotics as $a\to\infty$
is used in \cite{GMM}. We could use their Corollary 2.16
here
instead of referring to Sections 2,3.

We substitute $cz+d$ in (\ref{shin}) and put
$y(z)=w(cz+d)$, where
$$d=(a/3)^{1/2}>0,\quad c=(3d)^{-1/4}>0.$$
The result is
\begin{equation}
\label{aux}
-y^{\prime\prime}+(c^5z^3+z^2+\mu)y=0,
\end{equation}
where $\mu=c^2(\lambda+d^3-a d).$
Choosing the positive and negative imaginary rays
as our normalization rays $L_1$ and $L_2$,
we see that the normalization rays are
admissible 
in the sense of Theorem~\ref{thm2}. 
The Stokes complex of the binomial potential corresponding
to (\ref{aux}) is shown in Fig.~\ref{fig3}(a).
According to Theorem~\ref{thm2},
the spectrum of the problem (\ref{aux})
converges to the spectrum of the limit problem
\begin{equation}\label{aux2}
-y^{\prime\prime}+z^2y=-\mu y,\quad y(\pm i\infty)=0.
\end{equation}
This limit problem is equivalent to the self-adjoint
problem
$$-u^{\prime\prime}+z^2u=\mu u,\quad u(\pm\infty)=0$$
by the change of the variable $u(z)=y(iz)$.
Convergence of the spectrum implies convergence of
eigenfunctions uniform on compact subsets of the 
plane by Theorem \ref{thm3}. 
As $a$ varies from $0$ to $\infty$, we can choose
an eigenvalue $\lambda(a)$ which varies continuously,
and the corresponding eigenfunction that varies continuously,
and tends to an eigenfunction of (\ref{aux}).
In the process of continuous change the number of non-real
zeros of the eigenfunction cannot change because
eigenfunctions cannot have multiple zeros.
The conclusion of the theorem will now follow from
the known properties of zeros of eigenfunctions of
Hermitian boundary value problems, once we establish the
following 
\begin{lemma}\label{noescape} As $t=c^5\to 0$ in (\ref{aux})
the non-real
zeros of an eigenfunction cannot escape to infinity.
\end{lemma}
Notice that the real zeros of the eigenfunction do
escape to infinity, as the limit eigenfunction has
at most one real zero.
\begin{proof} Let $w_t$ be the eigenfunction
constructed in Theorem \ref{thm3} which depends 
continuously on $t$. Let $w_t^*$ be the Sibuya solution
corresponding to the positive ray.
Then $f_t=w_t/w^*_t$ is a real meromorphic solution
of the Schwarz equation and has asymptotic values
$\infty,0,b_t,\overline{b_t},0$ in the sectors $S_j$.
As $f_t\to f_0$, and the Schwarzian of $f_0$ is
of degree $2$, we conclude that $b=e^{i\beta}$
converges to the
real axis, and the Riemann surface of $f_t^{-1}$ must
converge in the sense of Caratheodory
\cite{C,Volk},
to
a Riemann surface with $4$ logarithmic branch points
which can lie only over $0,\infty,b_0$,
where $b_0\in\{\pm1\}$. 
To construct the cell decomposition
corresponding to this limit Riemann surface, we consider two
cases.
\vspace{.1in}

{\em Case 1.} $b_0=1.$
To describe the limit function,
we must replace in the original cell decomposition Fig.~\ref{3loops}
two loops corresponding to $b$, $\overline{b}$ with a single
loop around both of these points. This loop is shown by the
dotted line in Fig.~\ref{3loops} and its preimage is shown by the
dotted lines in Fig.~\ref{cubic-r}.
The original loops that separate $b$- and
$\overline{b}$- labeled faces from the face labeled $0$
must be removed. 
Performing this operation on
the cell decomposition $\Psi_n$ we see that the $1$-skeleton
breaks into infinitely many pieces. But there is only one piece that
has four unbounded faces and thus
can correspond to a meromorphic function whose
Schwarzian derivative is a polynomial of degree $2$.
This limit decomposition is shown in Fig.~\ref{cubic-c}.

\bigskip
\begin{figure}[ht]
\centering
\includegraphics[scale=0.6]{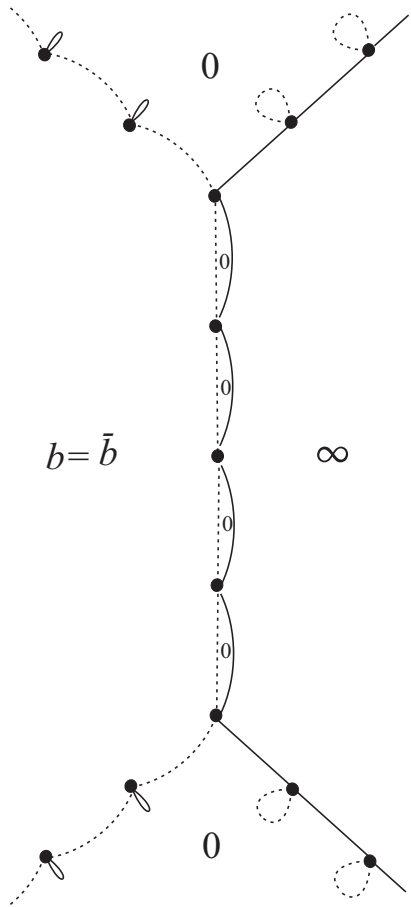}
\caption{Limit cell decomposition with $n=2$ (solid and dotted lines).}
\label{cubic-c}
\end{figure}

This time both solid and dotted lines represent the edges of
this decomposition.
We see that the
number of non-real zeros in the limit is the same
as it was before the limit.
\vspace{.1in}

{\em Case 2} $b_0=-1$. To analyze this case, we replace
the cell decomposition on Fig.~\ref{3loops} by the one in Fig.~\ref{3loopsa}

\bigskip
\begin{figure}[ht]
\centering
\includegraphics[scale=0.6]{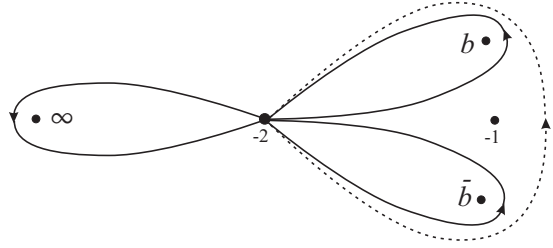}
\caption{ Another cell decomposition of the sphere (solid lines).}
\label{3loopsa}
\end{figure}

Now we want to find the preimage of the cell decomposition
in Fig.~\ref{3loopsa} under the same function $f$.
There are several ways to find this preimage.
Let us choose for convenience $b=i$, and
express the loops in Fig.~\ref{3loopsa}, in terms of the loops
in Fig.~\ref{3loops}, as elements of the fundamental group of
$\C\backslash\{0,i,-i\}$. 
We denote the loops around $b$,$\overline{b}$
in Fig.~\ref{3loops} by $\gamma_b,\gamma_{\overline{b}}$, and let
$\alpha$, $\beta$ be the upper and lower halves of the loop
around $\infty$,
so that $\gamma_\infty=\alpha\beta$ ($\alpha$ followed by $\beta$).
Let $\gamma^\prime_b$ and $\gamma^\prime_{\overline{b}}$
be the loops in Fig.~\ref{3loopsa}.
Then we have
$\gamma_\infty=\gamma'_\infty=\beta\alpha$ ($\alpha$ followed by $\beta$),
$\gamma'_b=\beta\gamma_b\beta^{-1}$,
$\gamma'_{\bar b}=\alpha^{-1}\gamma_{\bar b}\alpha$.
See Fig.~\ref{allloops}.

\begin{figure}[ht]
\centering
\includegraphics[scale=0.6]{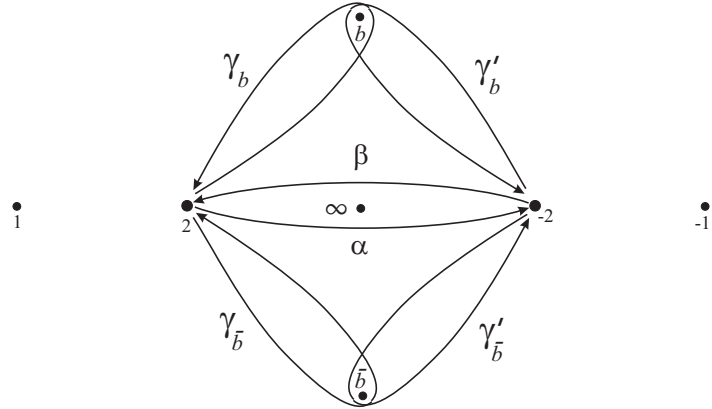}
\caption{Fig.~\ref{3loops} and Fig.~\ref{3loopsa} together.}
\label{allloops}
\end{figure}

These relations permit us to draw the preimages of the loops
$\gamma_\infty^\prime,\gamma_b^\prime,\gamma_{\overline{b}}^\prime$
in the $z$-plane.
The resulting picture is shown in Fig.~\ref{cubica-r}.

\begin{figure}[ht]
\centering
\includegraphics[scale=0.6]{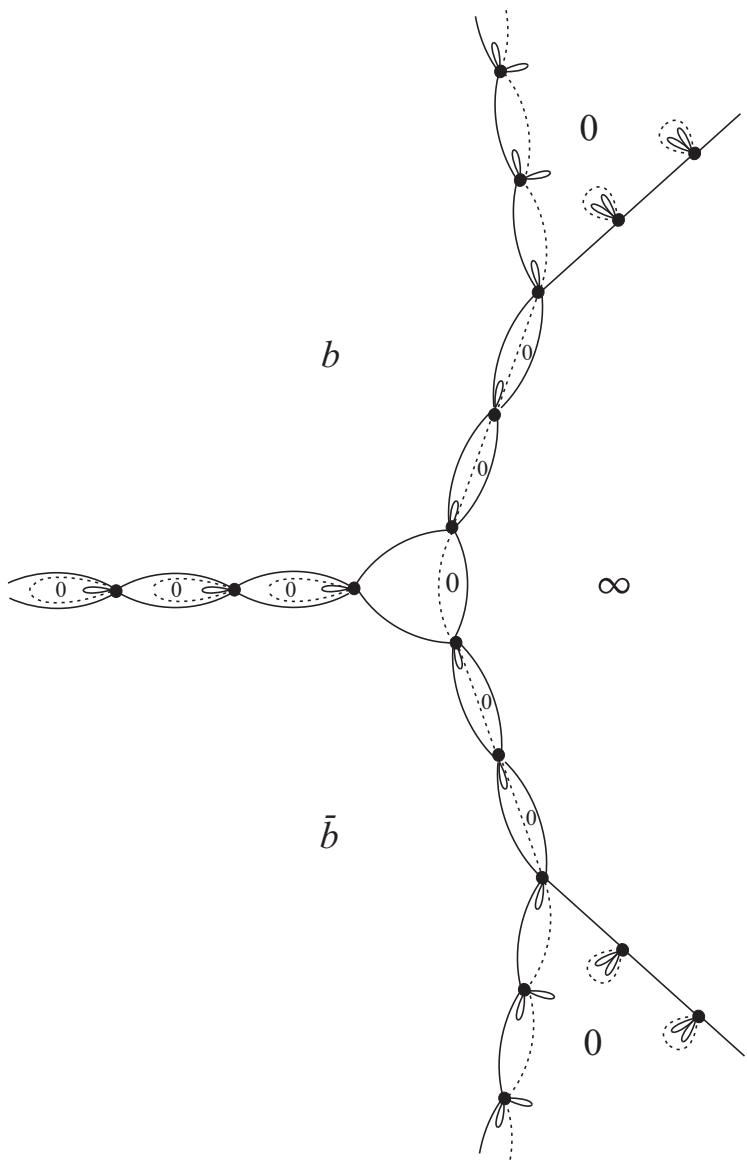}
\caption{Preimage of the cell decomposition
Fig.~\ref{3loopsa} (solid lines).}
\label{cubica-r}
\end{figure}

When $b\to-1$, we have a degeneration as before.
The corresponding cell decomposition is obtained by
replacing preimages of the loops $\gamma^\prime_b$
and $\gamma^\prime_{\overline{b}}$ by the preimages of the
dotted line. The resulting cell decomposition is shown in
Fig.~\ref{cubica-c}.

\begin{figure}[ht]
\centering
\includegraphics[scale=0.6]{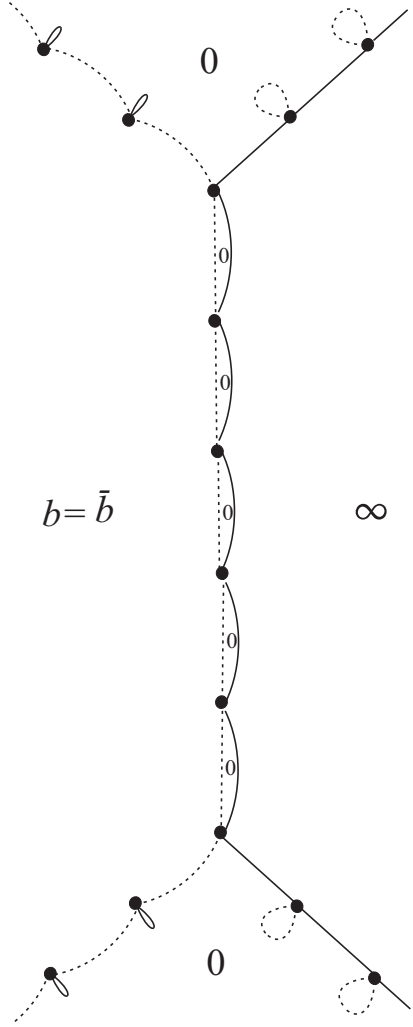}
\caption{ The limit cell decomposition of
Fig.~\ref{allloops} as $b\to-1$ (solid and dotted lines).}
\label{cubica-c}
\end{figure}

We see that the limit function still has $2n$ non-real zeros,
and one real zero.
This completes the proof of the lemma, and of
Theorem \ref{cubicth}.
\end{proof}

This proof shows that $\beta\to 0$ corresponds to the lower 
branch of $\Gamma_n$ while $\beta\to\pi$ corresponds to the upper
branch.

\end{proof}

\begin{theorem}\label{Hel-Ros1}
Let $P$ be a polynomial of degree $3$
such that equation (\ref{1-1}) has a solution with $2n$
non-real zeros.
Then (\ref{1-1}) can be transformed to an
equation of Theorem \ref{cubicth} with $(a,\lambda)\in\Gamma_n$
by a real affine change of the independent variable.
\end{theorem}

\begin{proof} By the results of Gundersen \cite{G2,G1},
all solutions
have infinitely many zeros, and the coefficients of 
$P$ are real. By a real affine change of the variable
we achieve that $P(z)=-z^3+az-\lambda$. As almost all
zeros are real, our solution must tend to zero
in both directions of the imaginary axis.
So $\lambda$ is an eigenvalue of the problem (\ref{shin}).
Let $w$ be a real eigenfunction and $w_1$ a real
solution of our equation that is linearly independent
of $w$. Then the ratio $f=w/w_1$ is a meromorphic function
which is a local homeomorphism, and has asymptotic
values $\infty,0,b,\overline{b},0$ in $S_0,\ldots,S_4$,
respectively. After a real affine change of the
independent variable, this function will belong
to the class $G$ defined in the proof of
Theorem~\ref{cubicth}.
\end{proof}
\vspace{.1in}

Now we state analogous results for quartic oscillators.
There are two different real two-parametric families
in which solutions with finitely many non-real
zeros can occur \cite{S2}.
\begin{equation}\label{bbmss}
-w^{\prime\prime}+(-z^4+az^2+cz+\lambda)w= 0,\quad w(\pm i\infty)=0.
\end{equation}
studied in \cite{BBMSS,DF2}, and
\begin{equation}\label{bb}
-w^{\prime\prime}+(z^4-2az^2+2mz+\lambda)w=0,
\quad w(te^{i\theta})\to 0,\quad
\theta=\pm\pi/3
\end{equation}
studied in \cite{BB}. Here $m\geq 1$ is an integer.
Problem \ref{bb} is quasi-exactly solvable,
which means that there are $m$ eigenfunctions
of the form $p(z)\exp(z^3/3-az)$,
where $p$ is a polynomial of degree $m-1$.
\vspace{.1in}
The families (\ref{bbmss}) and (\ref{bb}) are equivalent to
the PT symmetric families (\ref{BB}) and (\ref{BS1}-\ref{BS2})
of the Introduction
via the change of the independent variable
$z\mapsto iz$.

\begin{theorem}\label{quartic1} The real part of the
spectral locus
of (\ref{bbmss}) consists of disjoint
smooth connected analytic surfaces
$S_n,\; n\geq 0$, properly embedded in $\R^3$. For
$(a,c,\lambda)\in
S_n$, the eigenfunction has $2n$ non-real zeros.
Each of these surfaces is homeomorphic to a punctured disc.
Projection $\pi(a,c,\lambda)=(a,c)$ has the following properties:
It is a $2$-to-$1$ covering over some neighborhood
of the $a$-axis, and for $a>a_0$, the preimage
of every line $c=\const$ is compact and homeomorphic to a circle.
\end{theorem}

\begin{proof} 
We follow the same pattern as in the proof of Theorem \ref{cubicth}.
There are $6$ Stokes sectors, $S_0,\ldots,S_5$,
which we enumerate
anticlockwise, beginning from the sector
in the first quadrant.

If $f=w/w_1$
where $w$ is a real eigenfunction and $w_1$ is a real linearly independent solution
of the same equation, then $f$ has asymptotic values
$b_0,0,b_1,\overline{b_1},0,\overline{b_0}$ in the
sectors $S_0,\ldots,S_5$. Here $b_0\neq b_1$, and $b_0,b_1$
must belong to $\C\backslash\R$. 

If $c=0$, we can choose $w,w_1$
with the additional symmetry
with respect to the imaginary axis,
which gives $b_0=-\overline{b_1}$,
so $b_0$ and $b_1$ belong to the
same half-plane of $\C\backslash\R$.
The same situation persists for all real
$c$ because  $b_0,b_1$
depend continuously on $c$ and never cross the real line.
The real affine group acts on $f$ by post-composition;
this corresponds to the change
of normalization of $w$ and $w_1$. So we can always choose
the normalization so that $b_1=i$.

Notice that after this normalization
condition $c=0$ corresponds to 
$|b_0-i/2|=1/2.$ 
See Remark \ref{rem1} after the proof.

Consider the 
cell decomposition $\Phi$ of the Riemann sphere
(the range of $f$) shown in Fig.~\ref{4loops-r}.
It consists of one vertex at $\infty$ and four disjoint loops around
$\pm i$ and $b,\overline{b}$ that are interchanged by
the symmetry.

\begin{figure}[ht]
\centering
\includegraphics[scale=0.6]{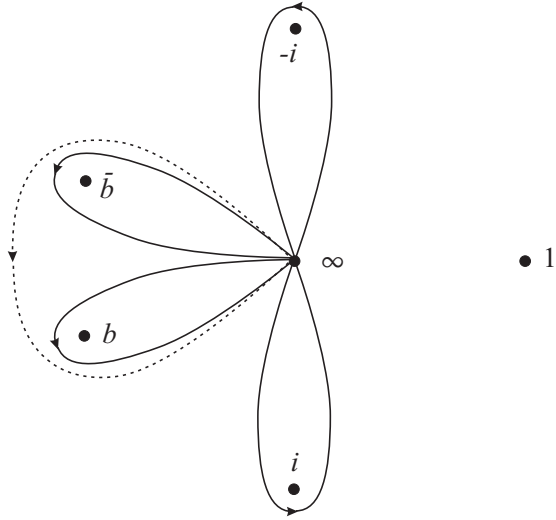}
\caption{Cell decomposition $\Phi$ of the sphere
(solid lines).}
\label{4loops-r}
\end{figure}

Now consider the cell decomposition $\Psi_n$
of the plane (with labeled
faces) shown in Fig.~\ref{quartic-r}. It is
locally similar to $\Phi$, and 
depends on one integer parameter $n\geq 0$ which is the number of
$0$-labeled faces between the adjacent ``ramification points''.

\begin{figure}[ht]
\centering
\includegraphics[scale=0.6]{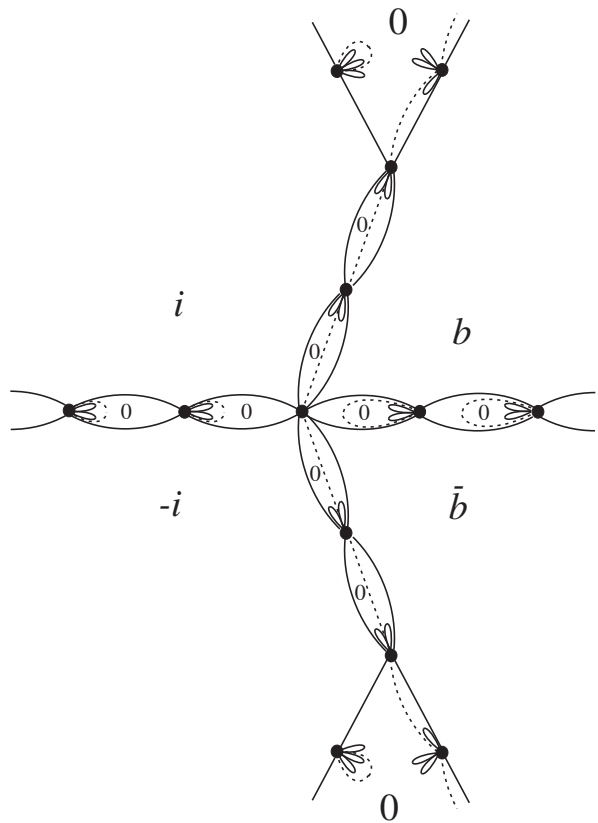}
\caption{Cell decomposition $\Psi_n$ for a quartic
with $n=2$ (solid lines).}
\label{quartic-r}
\end{figure}

As in Theorem~\ref{cubicth}, Nevanlinna theory gives for each $n\geq 0$ a family
$G_n$ of meromorphic functions $f$ which have $2n$ non-real zeros and
satisfy the Schwarz equation of the form
\begin{equation}\label{schwarz2}
\frac{f^{\prime\prime\prime}}{f^\prime}-\frac{3}{2}\left(
\frac{f^{\prime\prime}}{f^\prime}\right)^2=2(z^4-a z^2-cz-\lambda).
\end{equation}
with real $a,c,\lambda$.

Classification result for symmetric trees with $6$ faces in \cite{EG2}
ensures that all equations (\ref{1-1}) having a solution with infinitely
many real zeros
and $2n$ non-real zeros are equivalent
to equations which arise from our families $G_n$.

This also has an implication that there is ``no monodromy'' in our
families $G_n$: when $b$ traverses a loop around $i$, we return with
the same function $f$ we started with. Indeed, in the process of
continuous deformation the number of non-real zeros cannot change,
and there is only one suitable cell decomposition $\Psi_n$ for every $n$.

Thus our family $G_n$ is homeomorphic to a punctured disc.
Taking the coefficients $a,c,\lambda$
of the Schwarzian defines
an analytic embedding of $G_n$ to $\R^{3}$.
This is our surface
$S_n$. The surfaces are disjoint and
properly embedded for the same reasons
as in the proof of Theorem~\ref{cubicth}.

To study the shape of these surfaces $S_n$ in $\R^{3}$, we 
first notice that for $c=0$,
the eigenvalue problem obtained from (\ref{bbmss}) by rotation
$z\mapsto iz$ is Hermitian.
It follows that the intersection
with the plane
$S_n\cap \{ (a,c,\lambda):c=0\}$ consists of the disjoint graphs of two
analytic functions defined for all real $a$,
and that $\lambda_n(a,0)<\lambda_{n+1}(a,0)$.
Another simple property of the surface $S_n$ is that it is symmetric
with respect to change $c\mapsto -c$, which follows by changing $z\mapsto-z$
in the equation.

Now we study the asymptotic behavior of $S_n$ for $a\to+\infty$.
In the equation (\ref{bbmss}) we set
$z=\eps(\zeta-t),\; y(z)=w(\eps(\zeta-t))$, where $t$ 
satisfies 
\begin{equation}\label{ta}
a-6\eps^2t^2=0,\quad\mbox{and}\quad 4\eps^6t=1,
\end{equation}
and obtain
\begin{equation}\label{resc1}
-y^{\prime\prime}+(-\eps^6z^4+z^3+\alpha z+\mu)y=0,
\end{equation}
where
$$\alpha=4\eps^6t^3+2\eps^4at+c\eps^3,$$
and 
$$\mu=-\eps^6t^4+a\eps^4t^2-c\eps^3t+\eps^2\lambda.$$
Expressing $t$ and $a$ from equations (\ref{ta}) as functions
of $\eps$ and substituting the result to the expression
of $\alpha$ we obtain
\begin{eqnarray}\label{ac} 
a&=&(3/8)\eps^{-10},\quad c=\eps^{-3}\alpha-(1/4)\eps^{-15},\\
\lambda&=&-21\cdot2^{-8}\eps^{-20}+(\alpha/4)\eps^{-8}+\mu\eps^{-2}.
\label{lambda}
\end{eqnarray}
Consider the curves $\Gamma_n$ from Theorem~\ref{cubicth}.
It follows from their properties stated in Theorem~\ref{cubicth}, that for every
$n$, there exists $\alpha_n=\max\{-\alpha:(\alpha,\lambda)\in\Gamma_n$,
and $0<\alpha_n<\infty$.

Suppose that $\alpha<\alpha_n$,
and consider the curve in $(a,c)$-plane
pa\-ra\-met\-ri\-zed by (\ref{ac}).
Equation (\ref{resc1}) satisfies the conditions of Theorem~\ref{cubicth} of
Section 2 (the Stokes complex corresponding to this equation
is shown in Fig.~\ref{fig1a}, rotated by $90^\circ$),
and the sectors containing the normalizing rays are
stable.
We conclude that the spectrum of (\ref{resc1}) tends to the spectrum of
the cubic 
\begin{equation}
\label{cubic}
-y^{\prime\prime}+(z^3+\alpha z+\mu)y=0.
\end{equation}
The spectrum of the cubic with parameter $\alpha$ has at least one
eigenvalue $\mu^*$ which is real and such that the corresponding eigenfunction
has $2n$ non-real zeros. As $\mu^*$ is an isolated point of this spectrum,
and the spectrum of (\ref{resc1}) is symmetric with respect to the real axis,
we conclude that there is an eigenvalue $\mu$ of (\ref{resc1}) which
is real, and the corresponding eigenfunction has $2n$ non-real zeros.

To ensure that the number of non-real zeros does not change in the limit,
we make an argument similar to that in the proof of Theorem \ref{cubicth},
the degeneration of the cell decomposition $\Psi_n$ is shown
in Fig.~\ref{quartic-c}.

\begin{figure}[ht]
\centering
\includegraphics[scale=0.6]{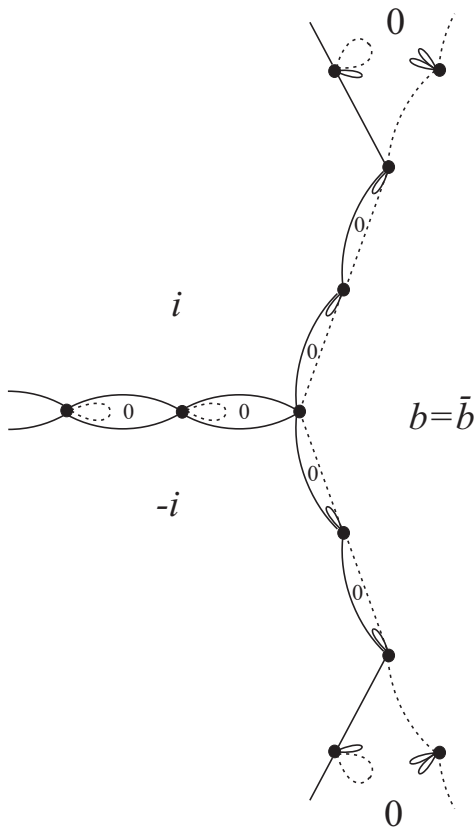}
\caption{ Limit cell decomposition with
$n=2$ (solid and dotted lines).}
\label{quartic-c}
\end{figure}

We conclude that projection of our surface $S$ contains a piece
of the curve (\ref{ac}) for $\eps\in (0,\eps_0)$.

Now suppose that $\alpha>\alpha_n$. We claim that there are no
points on $S_n$ with $a\to\infty$ and $(a,c)$ on the curve  (\ref{ac}).
Proving this by contradiction, we suppose that there is a sequence
$(a_j,c_j,\lambda_j)\in S_n$ such that $(a_j,c_j)$ belong to the
curve (\ref{ac}). Then Theorem~\ref{thm2} implies that
the sequence $\mu_j$ related to the $\lambda_j$ by (\ref{lambda}),
has the property that $\mu_j$ tends to a real eigenvalue $\mu^*$
of the cubic oscillator (\ref{cubic}). Then the corresponding eigenfunction
tends to an eigenfunction of the cubic with $2n$ non-real zeros.
This is a contradiction because $\alpha>\alpha_n$, so our claim
is proved.

So the projection of $S_n$ on the plane $(a,c)$ looks as a paraboloid
$9c^2-4a^3\leq 0$ when $a\to+\infty$.
\end{proof}

Fig.~\ref{trinh-quartic}, which is taken from Trinh's thesis shows a section
of the surfaces $S_n$ by the plane $a=-9$.
Similar pictures can be
seen in \cite{BBMSS,DF2}.

\begin{figure}[ht]
\centering
\includegraphics[scale=0.6]{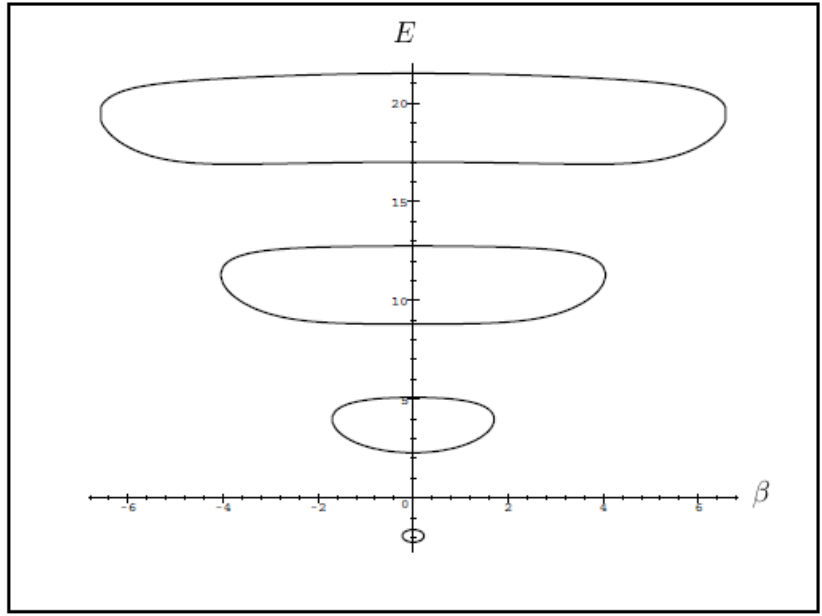}
\caption{Section of the surfaces $S_n,\; n=0,1,2,3$
by the plane $a=-9$.}
\label{trinh-quartic}
\end{figure}

Computational
evidence suggests that each $S_n$ has the shape of
an infinite funnel with a sharp end stretching towards
$a=-\infty$. This end probably corresponds to $b\to i$,
where $b$ is the asymptotic value as in Figs.~\ref{4loops-r}-\ref{quartic-r}.
$\lambda\to-\infty$ as $a\to-\infty$ on $S_n$ as the picture
in \cite{DF2} suggests.
For every real $a_0$ the section of $S_n$ by the plane
$a=a_0$ is an oval that projects on the $c$-axis $2$-to-$1$.
We only proved that this section is compact for $a$ large
enough. For $n=0,1,2,\ldots,$  the funnels $S_n$ are
symmetric with respect to $c\mapsto-c$, $S_{n+1}$ lies above
$S_n$ and $S_{n+1}$ is wider than $S_n$.

\begin{remark}\label{rem1}
In general, it is hard to say anything explicit on
the correspondence between the parameters $a,c$
in the potential and Nevanlinna parameter $b$.
Some information on this correspondence can be extracted
from symmetry and degeneration considerations.
In the beginning of the proof of Theorem we noticed
that the line $c=0$ corresponds to the circle
$|b-i/2|=1/2$. We can determine now the sign of $c$
for $b$ inside and outside this circle.
Degeneration used in the proof corresponds to
convergence of $b$ to a real non-zero point. Formula
(\ref{ac}) shows that $c<0$ when $\epsilon\to 0$.
So negative $c$ correspond to the exterior
of the circle and positive $c$ to its interior.
\end{remark}

Now we state the result about the second PT-symmetric family
of quartics.

\begin{theorem}\label{quartic2} The real QES part of the
spectral locus
of (\ref{bb}) consists of $\lceil m/2\rceil$ simple
disjoint analytic curves
$\Gamma_{m,n}^*$, $n=0,\ldots,\lceil m/2\rceil -1$,
properly embedded curves which for $a>0$ project onto
the ray $a>0$ $2$-to-$1$.
When $(a,\lambda)\in\Gamma_n^*$, the eigenfunction has $2n$ non-real zeros.
\end{theorem}

The proof is completely similar to the proof of Theorem \ref{cubicth}.
\vspace{.1in}

The problem of study of the whole real part
of spectral locus of (\ref{bb}), as a two-parametric family
with real $m$ seems quite interesting and challenging.
A picture of the spectral locus for $m=3$ can be seen in
\cite{BB}.

\begin{theorem}\label{Hel-Ros2} Every equation of
the form (\ref{1-1}) with polynomial $P$
of degree $4$ which has a solution with
$2n$ non-real zeros and
infinitely many real zeros
is equivalent by a real affine change of the independent
variable to equation (\ref{bbmss})
with $(a,c,\lambda)\in S_n$.

Every equation of the form (\ref{1-1}) with polynomial $P$
of degree $4$ which has a solution with finitely many zeros
is equivalent to (\ref{bb}) with $(a,\lambda)\in\Gamma_{m,n}^*$
by an affine change of the independent variable.
\end{theorem}

The proof is similar to that of Theorem \ref{Hel-Ros1}.

\vspace{.2in}



\end{document}